\documentclass [11pt]{article}

\usepackage{amsmath,xspace,amssymb,epsfig,fullpage}
\usepackage{algorithm}
\usepackage{algorithmic}
\usepackage{paralist}
\usepackage{amsthm}
\usepackage{stmaryrd}
\usepackage{subfig}
\usepackage{caption}
\usepackage{graphicx}
\usepackage{graphicx,color}
\DeclareCaptionType{copyrightbox}

\newtheorem{Thm}{Theorem}
\newtheorem{Lem}[Thm]{Lemma}
\newtheorem{Cor}[Thm]{Corollary}

\newtheorem{Def}{Definition}

\newcommand\mbR{\mbox{$\mathbb{R}$}}

\newcommand {\ie} {{i.e.}\xspace}
\newcommand {\st} {\textit{s.t.}\xspace}

\newcommand\nash{\mbox{\sf {Nash}}\xspace}
\newcommand\CE{\mbox{\sf {CE}}\xspace}

\newcommand\yes{\mbox{\sf {Yes}}\xspace}
\newcommand\no{\mbox{\sf {No}}\xspace}

\newcommand\vol{{\sf vol}}
\newcommand\region{{\sf region}}
\newcommand\vor{{\sf Vor}}
\newcommand\gvor{{\sf GenVor}}
\newcommand\extreme{{\sf extreme}}
\newcommand\atan{{\sf angle}}
\newcommand\SameVor{{\sf VorCheck}\xspace}
\newcommand\FurthestAlg{{\sf FurthestAlg}\xspace}

\newcommand\ap{a}

\newcommand{\Cylinders}{{\sf Cylinders}\xspace}
\newcommand{\AlgUnknownLP}{{\sc AlgLP}}
\newcommand{\ulp}{{\sf UnknownLP}\xspace}

\title{Solving Linear Programming with Constraints Unknown\\[0.2in]}

\author{
Xiaohui Bei\footnote{Nanyang Technological University, Singapore. Email: beixiaohui@gmail.com.} \and Ning Chen\footnote{Nanyang Technological University, Singapore. Email: ningc@ntu.edu.sg.} \and Shengyu Zhang\footnote{The Chinese University of Hong Kong, Hong Kong. Email: syzhang@cse.cuhk.edu.hk.}
}

\date{}

\begin{document}

\maketitle

\begin{abstract}
What is the value of input information in solving linear programming? The celebrated ellipsoid algorithm tells us that the full information of input constraints is not necessary; the algorithm works as long as there exists an oracle that, on a proposed candidate solution, returns a violation in the format of a separating hyperplane. Can linear programming still be efficiently solved if the returned violation is in other formats?



Motivated by some real-world scenarios, we study this question in a trial-and-error framework: there is an oracle that, upon a proposed solution, returns the \emph{index} of a violated constraint (with the content of the constraint still hidden). When more than one constraint is violated, two variants in the model are investigated. (1) The oracle returns the index of  a ``most violated'' constraint, measured by the Euclidean distance of the proposed solution and the half-spaces defined by the constraints. In this case, the LP can be efficiently solved (under a mild condition of non-degenerency). (2) The oracle returns the index of an arbitrary (\ie, worst-case) violated constraint. In this case, we give an algorithm with running time exponential in the number of variables. We then show that the exponential dependence on $n$ is unfortunately necessary even for the query complexity. These results put together shed light on the amount of information that one needs in order to solve a linear program efficiently.

The proofs of the results employ a variety of geometric techniques, including McMullen's Upper Bound Theorem, the weighted spherical Voronoi diagram, and the furthest Voronoi diagram. In addition, we give an alternative proof to a conjecture of L{\' a}szl{\' o} Fejes T{\' o}th on bounding the number of disconnected components formed by the union of $m$ convex bodies in $\mathbb{R}^n$. Our proof, inspired by the Gauss-Bonnet Theorem in global differential geometry, is independent of the old one by Kovalev \cite{Kov88} and reveals more clear insights into the problem and the bound.
\end{abstract}

\setcounter{page}{0}\thispagestyle{empty}
\newpage


\section{Introduction}\label{sec:intro}

Solving linear programming (LP) is a central question studied in operations research and theoretical computer science.
The existence of efficient algorithms for LP is one of the cornerstones of a broad class of designs in, for instance, approximation algorithms and combinatorial optimization.
The feasibility problem of linear programming asks to find an $x\in \mathbb{R}^n$ to satisfy a number of linear constraints $Ax > b$.
Some previous algorithms, such as the simplex and interior point algorithms, assume that the constraints are explicitly given.
In contrast, the ellipsoid method is able to find a feasible solution even without full knowledge of the constraints.
This remarkable property grants the ellipsoid method an important role in many theoretical applications.


A central ingredient in the ellipsoid method is an oracle that, for a proposed (infeasible) point $x\in \mathbb{R}^n$, provides a violation that separates $x$ and the feasible region of the LP in the format of a hyperplane.
Such a separation oracle captures situations in which the input constraints are unavailable or cannot be accessed affordably, and the available information is from separating hyperplanes for proposed solutions.
A natural question is what if the feedback for a proposed solution is not a separating hyperplane.
Aside from theoretical curiosity, the question relates to practical applications, where the acquired violation information is actually rather different and even more restricted and limited.

Transmit power control in cellular networks has been extensively studied in the past two decades, and the techniques developed have become foundations in the CDMA standards in today's 3G networks. In a typical scenario, there are a number of pairs of transmitters and receivers, and the transmission power of each transmitter needs to be determined to ensure that the signal is strong enough for the target receiver, yet not so strong that it interferes with other receivers. This requirement can be written as an LP of the form $Ax>b$, where each constraint $i$ corresponds to the requirement that the Signal to Interference Ratio (SIR) is no less than a certain threshold. In general the power control is a well-known hard problem (except for very few cases, such as power minimization~\cite{FM93}); a major difficulty is that matrix $A$ depends mainly on the ``channel gains'', which are largely unknown in many practical scenarios~\cite{CHLT08}.
Thus the LP $Ax>b$ needs to be solved despite the unavailability of $(A,b)$. What is available here is that the system can try some candidate solution $x$ and observe violation information (namely whether the SIR exceeds the threshold). The system can then adjust and propose new solutions until finally finding an $x$ to satisfy $Ax>b$.
There are more examples in other areas (e.g., normal form games and product design and experiments \cite{Mon08}) with input information hidden. In these examples,
for any unsatisfied proposed solution, only certain salient phenomena of violation (such as signal interference) are exhibited, which give \emph{indices} of violated constraints but not their contents.
 With so little information obtained from violations, is it still possible to solve linear programming efficiently?
We attempt to answer that question in the present paper. Our work aims to address the value of input information in solving LP, and can hopefully help to deepen our understanding of the following general question.
\begin{center}
  \begin{quote}
\emph{What is the least amount of input information, in what format, that one needs to solve a linear program efficiently?}
   \end{quote}
\end{center}

\noindent In this paper, we study the above question by testing both sides of the boundary.


%
%
%


\subsection{Our Model and Results}

Our model is defined as follows. In an LP $Ax>b$, the constraints $a_ix>b_i$ are hidden to us. We can propose candidate solutions $x\in \mathbb{R}^n$ to a \emph{verification oracle}\footnote{The verification oracle is simply a means of determining whether a solution is feasible. It arises from the nature of LP as shown from the foregoing examples. For infeasible solutions, the feedback is a signaled violation.}. If $x$ satisfies $Ax > b$, then the oracle returns \yes and the job is done. If $x$ is not a feasible solution, then the oracle returns the index of a violated constraint. The algorithm continues until it either finds a feasible solution or concludes that no feasible solution exists. The algorithm is adaptive in the sense that future queries may depend on the information returned during previous queries.
We focus only on the feasibility problem, to which an optimization LP can be transformed by a standard binary search.

Note that when the proposed solution is not feasible, the oracle returns only the \emph{index} $i$ of a violation rather than the constraint $a_ix>b_i$ itself. We make this assumption for two reasons. First, consistent with the aforementioned examples, we are often only able to observe unsatisfactory phenomena (such as a strong interference in the power control problem). However, the exact reasons (corresponding to the content of violated constraints) for these problems may still be unknown. Second, as our major focus is on the value of information in solving linear programming, a weaker assumption on the information obtained implies stronger algorithmic complexity results. Indeed, as will be shown, in some settings efficient algorithms exist even with this seeming deficit of information.

For a proposed solution $x$, if there are multiple violated constraints, the oracle returns the index of one of them\footnote{It is also natural to consider the case where the oracle returns the indices of all violated constraints. That model turns out to be so strong as to make the linear program easily solvable. We study oracles returning only one index to emphasize that even given such limited information, efficient algorithms exist in some settings.}. This raises the question of which violation the oracle returns, and two variants are studied in this paper. In the first one, the oracle gives more information by returning the index of a ``most violated'' constraint, where the extent of a violation is measured by the Euclidean distance of the proposed solution $x$ and the half-space defined by the constraint. This oracle, referred to as the \emph{furthest oracle}, attempts to capture the situation in which the first violation that occurs or is observed is usually the most severe one. 
The second variant follows the tradition of worst-case analysis in theoretical computer science, and makes no assumption about the returned violation. This oracle is referred to as the \emph{worst-case oracle}.

We will denote by \ulp the problem of solving LP with unknown constraints in the above model. In either oracle model, 
the time complexity is the minimum amount of time needed for any algorithm to solve the \ulp problem, where each query, as in the standard query complexity, costs a unit of time.

Our results are summarized below. In a nutshell, when given a furthest oracle, a polynomial-time algorithm exists to solve LP (under a mild condition of non-degeneracy). On the other hand, if only a worst-case oracle is given, the best time cost is exponential in $n$, the number of variables. Note that it is efficient when $n$ is small, a well-studied scenario called \emph{fixed-dimensional LP}. The exponential dependence on $n$ is unfortunately necessary even for the query complexity. This lower bound, when combined with the positive result for the furthest oracle case, yields an illustration of the boundary of tractable LP.

\paragraph{Furthest oracle.}
The worst-case oracle necessitates an exponential time complexity, but in some practical applications failed trials may reveal more violation information. For instance, in the power control problem all of the distances between the proposed solution point to the half-spaces of violated constraints can be estimated and reported. Is this additional information greatly helpful in reducing the computational cost? In general, what is the least amount of information about violations needed to solve an LP efficiently? Compared to providing all distances, the furthest oracle reveals only a small amount of extra information by returning the index of a most violated constraint. However, this turns out to be sufficient to admit a polynomial-time algorithm.

As mentioned earlier, returning the indices of all violated constraints makes the model strong enough to admit efficient algorithms. Since our goal is to understand the boundary of tractability, it is desirable to have a model as weak as possible in which tractability is still maintained. The furthest oracle is defined for this purpose: compared to providing all violations, the furthest oracle only reveals one of the indices $i$ among the maximizers in $\max_i (b_i - \langle a_i, x\rangle)/\|a_i\|$. Despite this small amount of information, surprisingly, the furthest oracle turns out to be informative enough to admit a polynomial-time algorithm under a mild condition.

\begin{Thm}
The \ulp problem can be solved in time polynomial in the input size in the furthest oracle model, provided that the input is non-degenerate\footnote{The exact definition of non-degeneracy is given in Section \ref{section-furthest}. The condition is mild; actually a random perturbation on inputs yields non-degeneracy, thus the theorem implies that the smooth complexity is polynomial.}.
\end{Thm}

The main idea of the algorithm design is as follows. Instead of searching for a solution directly, we consider the unknown matrix $A$ and vector $b$ as a degenerate polyhedron in $\mathbb{R}^{m(n+1)}$, and use the ellipsoid method to find $(A,b)$. In each iteration we consider the center $(A',b')$ of the current ellipsoid in $\mathbb{R}^{m(n+1)}$, and aim to construct a separating hyperplane between $(A,b)$ and $(A',b')$ through queries to the furthest oracle. The main difficulty lies in the case when $(A',b')$ is infeasible, in which a separating hyperplane cannot be constructed explicitly. It can be observed that upon a query $x$, with the help of the furthest oracle, the information returned from the oracle has a strong connection to the Voronoi diagram. Specifically, if $x$ is not a feasible solution, then the returned index is always the furthest Voronoi cell that contains $x$. We can manage to compute the Voronoi diagram, but this does not uniquely determine the constraints that define the LP. To handle this difficulty, we give a sufficient and necessary characterization reducing the input LP to that of a new and homogeneous LP, for which the constraints can be identified using the structure of a corresponding weighted spherical closest Voronoi diagram.



\paragraph{Worst-case oracle.}
Recall that the worst-case oracle may return the index of an arbitrary violation. In this case, we first establish the following upper bound which is exponential in the number of variables only.

\begin{Thm}
The \ulp problem with $m$ constraints, $n$ variables, and input size $L$ can be deterministically solved in time $(mnL)^{poly(n)}$. In particular, the algorithm is of polynomial time for constant dimensional LP (\ie constant number of variables).
\end{Thm}

At the heart of the efficiency guarantee of our algorithm is a technical bound of $\sum_{i=0}^n \binom{m}{i}$ on the number of ``holes'' formed by the union of $m$ convex bodies in $\mathbb{R}^n$. This bound was first conjectured by L{\' a}szl{\' o} Fejes T{\' o}th. The 2-dimensional case was proved by Katona~\cite{Kat77} in 1977, based on an analysis of the shape of the convex sets, and the general case was proved by Kovalev~\cite{Kov88} in 1988, by induction on dimension. 
We give an independent and completely different proof, which is simpler and does not rely on induction. Compared to the previous proofs, ours reveals the nature of the problem and exhibits a clear and simple reason for the bound to hold. (One can see clearly from our proof where each summand in the bound comes from.) The main idea and some technical tools in our proof are inspired by the high-dimensional Gauss-Bonnet theorem, the most important theorem in global differential geometry.
A key concept needed in our proof is a properly defined high-dimensional ``exterior angle'', which connects the convex bodies and the ``holes'' at every boundary point.
Our exterior angle differs from the standard one 
by Banchoff~\cite{Ban67} by dropping all low-dimensional terms, but only in this way does it yield a critical identity that we need: the integral of all exterior angles of any bounded set, convex or non-convex, is 1.

The above theorem implies a polynomial time algorithm when the number $n$ of variables is a constant. This is a well-studied scenario, called \emph{fixed dimensional LP} in which $n$ is much smaller than the number of constraints $m$; see \cite{Kal92,Cla95,MSW96,Meg84,Dye86} and the survey~\cite{DMW04}.

On the other hand, a natural question is whether the exponential dependence is necessary; at the very least, can we improve the bound to $poly(m,n) + 2^{poly(n)}$, as Matousek et al.~\cite{MSW96} have done, which is still polynomial when $n$ is slowly growing as some $poly\log(m)$? Unfortunately, the next lower bound theorem indicates that this is impossible.

\begin{Thm}
Any \emph{randomized} algorithm that solves the \ulp problem with $m$ constraints and $n$ variables needs $\Omega\big(m^{\lfloor n/2\rfloor}\big)$ queries to the oracle, regardless of its time cost.
\end{Thm}

The lower bound implies that our algorithm, although of an exponential complexity, is close to optimal.
Our proof of the lower bound uses the dual of the seminal Upper Bound Conjecture, proved by McMullen~\cite{McM70,MS71}, which gives a tight upper bound on the number of faces in an $n$-dimensional cyclic polytope with $m$ vertices.

It is worth comparing the exponential hardness of \ulp with the complexities of \nash and \CE, the problems of finding a Nash or correlated equilibrium in a normal-form game, in the trial-and-error model. In our previous work~\cite{BCZ12}, we presented algorithms with \emph{polynomial} numbers of queries for \nash and \CE with unknown payoff matrices in the model with worst-case oracle\footnote{An algorithm proposes a candidate equilibrium and a verification oracle returns the index of an arbitrary better response of some player as a violation.}. \nash and \CE can be written as quadratic and linear programs, respectively, but why is the general \ulp hard while the unknown-input \nash and \CE are easy (especially when all are given unlimited computational power)? The most critical reason is that in normal-form games, there \emph{always exists} a Nash and a correlated equilibrium, but a general linear program may not have feasible solutions. Indeed, if a feasible solution is guaranteed to exist (even for only a random instance), such as when the number of constraints is no more than that of variables, then an efficient algorithm for \ulp does exist: see Appendix~\ref{section-smallconstraint}. (In our algorithms for \ulp, the major effort is devoted to handling infeasible LP instances.) It is interesting to see that the {\em solution-existing} property plays a fundamental role in developing efficient algorithms.

\subsection{Related Work}

A considerable body of work has studied the value of information in various domains. We consider algorithmic computation of linear programming from the perspective of available information.
Papadimitriou and Yannakakis~\cite{PY93} also studied solving linear programming with matrix unknown. However, their setting is very different from ours. They studied a specific class of linear programs, $Ax\le 1$ and $x\ge 0$ where the matrix $A\ge 0$, and considered a set of decision-makers who hold each of the variables and only know all of the constraints containing the variable. In addition, they studied the problem in the distributed decision-making setting, and focused on designing distributed algorithms with the objective of maximizing $\sum_i x_i$.
Ryzhov and Powell~\cite{RP12} studied information collection in linear programming, but their unknown is the coefficients of the objective function.

In our previous work~\cite{BCZ12}, we studied the trial-and-error approach to finding a feasible solution for a search problem with unknown input and a verification oracle for a number of combinatorial problems, such as stable matching, SAT, group and graph isomorphism, and the Nash equilibrium.
However, to bypass the computational barrier for some problems (e.g., SAT), \cite{BCZ12} equipped an algorithm with a separate computation oracle, whereas in the present paper we only have the verification oracle.
In addition, we consider not only the worst-case oracle but also a natural furthest oracle.
Finally, our major focus is on the algorithm design in solving the \ulp problem, but the main consideration of~\cite{BCZ12} is the relative complexity of solving a search problem with an unknown input compared to that with a known input.


\newcommand{\V}{\mbox{\sf {V}}\xspace} 

\section{Preliminaries}
Consider the following linear program (LP):
$Ax > b,$
where $A=(a_{ij})_{m\times n}\in \mathbb{R}^{m\times n}$ and $b=(b_1,\ldots,b_m)^T\in \mathbb{R}^m$.
The {\em feasibility} problem asks to find a feasible solution $x\in \mathbb{R}^n$ that satisfies $Ax > b$
(or report that such a solution does not exist). Equivalently, this is to find a point $x\in \mathbb{R}^n$ that satisfies $m$ linear constraints
$\{a_ix > b_i : i\in [m]\}$, where each $a_i=(a_{i1},\ldots,a_{in})$.

In the {\em unknown-constraint LP feasibility} problem, denoted by \ulp, the coefficient matrix $A$ and the vector $b$ are unknown to us, but we still need to determine whether the LP has a feasible solution and find one if it does.
The way of solution finding is through an adaptive interaction with a {\em verification oracle}:
We can propose candidate solutions $x\in \mathbb{R}^n$. If a query $x$ is indeed a feasible solution,
the oracle returns \yes and the job is done. Otherwise, the oracle returns an index $i$ satisfying $a_ix\le b_i$, i.e., the index of a violated constraint.
Note that we know only the {\em index} $i$, but not $a_i$ and $b_i$, the content of the constraint.
In addition, if multiple constraints are violated, only the index of \emph{one} of them is returned.

In the present paper, we study the computational complexity of solving the \ulp problem. As in the standard complexity theory with oracles, we assume that each query to the oracle takes unit time. We will analyze the complexity for two types of oracles: the \emph{worst-case oracle} which can return an arbitrary index among those violated constraints 
(Section~\ref{section-worst}), and the \emph{furthest oracle} which returns the index of a ``most" violated constraint (Section~\ref{section-furthest}). 

\medskip \noindent
\textbf{Input size and solution precision.}
A clarification is needed for the size of the input. 
Since the input LP instance $(A,b)$ is unknown, neither do we know its binary size.
To handle this issue, we assume that we are given the information that there are $m$ constraints\footnote{Indeed, the number of constraints can be unknown to us as well: In an algorithm, we only need to track those violated constraints that have ever been returned by the oracle.},
$n$ variables, and the binary size of the input instance $(A,b)$ is at most $L$. Note that $L$ is $O(mn\log(N))$, where $N$ is the maximum entry (in abstract value) in $A$ and $b$. We say that an algorithm solves \ulp efficiently if its running time is $poly(m, n,L)$.

Given an LP with input size $L=O(mn\log(N))$, it is known~\cite{Kha79} that if the LP has a feasible solution, then there is one whose numerators and denominators of all components are bounded by $(nN)^n$. Hence, an alternative way to describe our assumption is that, instead of knowing the input size bound $L$, there is a required precision for feasible solutions.
That is, we only look for a feasible solution in which the numerators and denominators of all components are bounded by the required precision.
These two assumptions, i.e., giving an input size bound and giving a solution precision requirement, are equivalent,
and it is necessary to have one of them in our algorithms.\footnote{Otherwise, we may not be able to distinguish between cases when there are no feasible solutions (e.g., $x>0,x<0$)
and when there are feasible solutions but the feasible set is very small (e.g., $x > 0, x < \epsilon$).
For any queried solution $y>0$, the oracle always returns that the second constraint is violated. However, we cannot distinguish whether it is $x<0$ in the first LP or $x<\epsilon$ in the second LP,
as $\epsilon$ can be arbitrarily small and we have no information on how small it is.}
In the rest of the paper, we will use the first one, the input size bound, to analyze the running time of our algorithms.

%
%

\medskip \noindent
\textbf{Geometric background.}
The geometric concepts, notation and facts that we will use are summarized as follows. The unit sphere in $\mbR^n$ is denoted by $S^{n-1} = \{x\in \mbR^n: \|x\|= 1\}$, where, throughout this paper, $\|\cdot\|$ refers to the $\ell_2$-norm.

\begin{Def}
  A set $C\subseteq \mathbb{R}^n$ is a {\em convex cone} if for any $x, y \in C$
  and any $\alpha, \beta > 0$, $\alpha x + \beta y$ is also in $C$.
  The \emph{normalized volume} (also called volumetric modulus) of a convex cone $C$ is defined as the ratio
  $$v(C) = \frac{\vol_{n}(C \cap B^{n})}{\frac{1}{2}\cdot \vol_{n}(B^{n})}$$
  where $B^{n}$ is the closed unit ball in $\mathbb{R}^n$ and $\vol_n$ refers to the $n$-dimensional volume.
\end{Def}

\begin{Def}
  For any set $C\in \mbR^n$, its \emph{polar cone} $C^{*}$ is the set
  $$C^{*} = \big\{y \in \mathbb{R}^n : \langle x, y \rangle \leq 0, \forall x \in C\big\}.$$
\end{Def}

\begin{Def}
  For any point set $P$, its \emph{convex hull} $conv(P)$
  is the intersection of all convex sets that contain $P$.
   In particular, for any points $p_1, p_2, \ldots, p_m \in \mathbb{R}^n$,
   $$conv\big(\{p_1,p_2, \ldots, p_m\}\big) = \left\{\sum_{i=1}^{m}\lambda_ip_i
   : \lambda_i \geq 0, \sum_{i=1}^m\lambda_i = 1\right\}.$$
\end{Def}

We will use the following technical lemmas.

\begin{Lem}[\cite{San54}] \label{lem:polar}
  Let $C_1, C_2, \ldots, C_k$ be $k$ closed convex cones, then $(\bigcap_i C_i)^* = conv(\bigcup_i C_i^*)$.
\end{Lem}

It was shown in~\cite{PS98} (Lemma~8.14) that if an LP has a feasible solution, then the set of solutions within the ball $\big\{x\in \mathbb{R}^n : \|x\|\le n2^L\big\}$ has volume at least $2^{-(n+2)L}$. Given this lemma, we can easily derive the following claim.

\begin{Lem}\label{lem:feasiblesize}
  If a linear program $Ax > 0$ 
  has a feasible solution, then the feasible region is a convex cone in $\mathbb{R}^n$ and
  has normalized volume no less than $2^{-(2n+3)L}$.
\end{Lem}


\section{Furthest Oracle}\label{section-furthest}

In this section, we will consider the \ulp problem $Ax> b$ with 
the \emph{furthest oracle}, formally defined as follows. For a proposed candidate solution $x$, if $x$ is not a feasible solution, instead of returning the index of an arbitrary (worse case) violated constraint, the oracle returns the index of a ``most violated" constraint, measured by the Euclidean distance from the proposed solution $x$ and the half-space defined by the constraint. More precisely, the oracle returns the index of a constraint which, among all $i$ with $\langle a_i, x \rangle \leq b_i$, maximizes $\frac{b_i - \langle a_i, x \rangle}{\|a_i\|}$, the distance from $x$ to the half-space $\{z\in \mathbb{R}^n: \langle a_i, z \rangle \geq b_i\}$.
If there are more than one maximizer, the oracle returns an arbitrary one.

Compared to the worse-case oracle, the furthest oracle reveals more information about the unknown LP system, and
indeed, it can help us to derive a more efficient algorithm. Our main theorem in this section is the following.

\begin{Thm}
The \ulp problem $Ax>b$ with a non-degenerate matrix $A$ in the furthest oracle model can be solved in time polynomial in the input size.
\end{Thm}

We call a matrix $A=(a_1,\ldots,a_m)^T$ {\em non-degenerate} if for each point $p\in S^{n-1}=\{x\in \mathbb{R}^n : \|x\|=1\}$, at most $n$ points in $\big\{\frac{a_1}{\|a_1\|},\ldots,\frac{a_m}{\|a_m\|}\big\}$ have the same spherical distance to $p$ on $S^{n-1}$.
This assumption is with little loss of generality; it holds for almost all real instances and can be derived easily by a small perturbation.

Note that in the worst-case oracle setting, we can easily reduce the general LP $Ax>b$ to $Ax-by>0$ by adding a new variable $y$.
However, the same trick does not apply to the furthest oracle setting. This is because for a given query, the furthest violated constraint in $Ax>b$ can be different from that in $Ax-by>0$. 
Next we will first describe our algorithm for the special case $Ax>0$, then generalize the algorithm to the $Ax>b$ case.
The formal proof of the algorithm is deferred to Appendix~\ref{appendix-furthest-formal}.

\subsection{Algorithm Solving {\large $Ax>0$}}\label{sec:Ax>0}

We assume without loss of generality that $\|a_i\| = 1$ for all $i$. Furthermore, we can also always propose points in $S^{n-1}$ for the same reason.

\medskip \noindent
\textbf{Ellipsoid method and issues.}
The main approach of the algorithm is to use the ellipsoid method to find the unknown matrix $A = (a_{ij})_{m\times n}$,
which can be viewed as a point in the dimension $\mathbb{R}^{mn}$, i.e., a degenerate polyhedron in $\mathbb{R}^{mn}$.
Initially, for the given input size information $m,n$ and $L$, we choose a sufficiently large ellipsoid that contains the candidate region of $A$, and pick the center $A'\in \mathbb{R}^{mn}$ of the ellipsoid. To further the ellipsoid method, we need a hyperplane that separates $A'$ from the true point $A$.

Consider the linear system $A'x > 0$. If it has a feasible solution $x$, then $\{x:A'x>0\}$ is a full-dimensional cone.  We query an $x$ in this cone to the oracle.
If the oracle returns an affirmative answer, then $x$ is a feasible solution of $Ax > 0$ as well, and the job is done.
Otherwise, the oracle returns an index $i$, meaning that $\langle a_i, x\rangle\le 0$.
Hence, we have $\langle a'_i, x\rangle >0 \ge \langle a_i,x\rangle$, which defines a separating hyperplane between $A$ and $A'$
(note that we know the information of $A'$ and $x$).
Thus, we can cut the candidate region of $A$ by a constant fraction and continue with the ellipsoid method.



Note that there is a small issue: In our problem, the solution polyhedron degenerates to a point $A\in \mathbb{R}^{mn}$ and has volume 0.
As the input $A$ is unknown, we cannot use the standard approach in the ellipsoid method to introduce a positive volume for the
polyhedron by adding a small perturbation. This issue can be handled by a more involved machinery developed by Gr{\"o}tschel,
Lov{\'a}sz, and Schrijver~\cite{GLS84,GLS88}, which solves the strong nonemptiness problem for well-described polyhedra given by a strong separation oracle, as long as a \emph{strong separation oracle} exists. In the algorithms described below, we will construct such oracles, thereby circumventing the issue of perturbation of the unknown point $A$.
The same idea has been used in \cite{BCZ12} to find a Nash equilibrium when the payoff matrix is unknown and degenerates to a point in a high-dimensional space. More discussions refer to~\cite{GLS84,GLS88,BCZ12}.



The main difficulty is when the LP $A'x > 0$ is infeasible. 
In the following part of this section we will discuss how to find a proper separating hyperplane in this case.

\medskip \noindent
\textbf{Spherical (closest) Voronoi diagram.}
Note that $Ax > 0$ is equivalent to $-Ax < 0$, and $i$ minimizes $\langle a_i, x \rangle$ if and only if it maximizes $\langle -a_i, x \rangle$. In the rest of this subsection, for  notational convenience, we use $x\in S^{n-1}$ to denote a proposed solution point, and let $y=-x$.
Since the distance from a proposed solution $x$ to a half-space
$\{z\in \mathbb{R}^n :  \langle a_i, z \rangle \geq 0\}$ is $-\langle a_i, x \rangle = \langle a_i, y \rangle$, the oracle returns us an index
$i \in \arg\max_i \big\{\langle a_i, {y} \rangle : \langle a_i, y \rangle \ge 0\big\}$
if $x$ is not feasible. Note that $\|z - a_i\| \leq \|z - a_j\|$ if and only if $\langle a_i, z \rangle \geq \langle a_j, z \rangle$ for any $z\in S$;
thus, $\langle a_i, y \rangle$ is closely related to the distance between $a_i$
and $y$ on $S^{n-1}$. That is, the oracle actually provides information about the closest Voronoi diagram of $a_1, \ldots, a_m$ on $S^{n-1}$.

Recall that the (closest) {\em Voronoi diagram} (also called Dirichlet tessellation) of a set of points $\{a_i\}_i$ in a space $S$ is
a partition of $S$ into cells, such that each point $a_i$ is associated with the cell $\{z\in S : d(z,a_i) \leq d(z,a_j), \forall j\}$, where $d$ in our case is the spherical distance on $S^{n-1}$.
We denote by $\vor$ the spherical (closest) Voronoi diagram of the points $a_1,\ldots,a_m$ on $S^{n-1}$ and denote by $\vor(i)$ the cell in the diagram associated with $a_i$, i.e.,
\begin{eqnarray}
  \vor(i) &=& \big\{z\in S^{n-1} :  \langle a_i,z \rangle \ge \langle a_j,z \rangle, \ \forall j\in [m]\big\} \label{eq:vor} \\
  &=& \big\{z\in S^{n-1} :  \|z - a_i\| \leq \|z - a_j\|, \ \forall j\in [m]\big\}. \nonumber
\end{eqnarray}
If the oracle returns $i$ upon a query $x=-y\in S^{n-1}$, then $y\in \vor(i)$.

\medskip \noindent
\textbf{Representation.} Note that for a general (spherical) Voronoi diagram formed by $m$ points, it is possible that some of its cells contain exponential number of vertices, which is unaffordable for our algorithm. However, in the $H$-representation of a convex polytope, every cell can be represented by at most $m$ linear inequalities, as shown in Formula~(\ref{eq:vor}). In the following, we will see that the information of these linear inequalities is sufficient to implement our algorithm efficiently.


\medskip \noindent
\textbf{Weighted spherical (closest) Voronoi diagram.}
For the presumed matrix $A'$, note that it can be an arbitrary point in the space $\mbR^{mn}$ and may not necessarily fall into $S^{n-1}$.
Our solution is to consider a {\it weighted spherical Voronoi diagram}, denoted by $\vor'$, of points $\frac{a'_1}{\|a'_1\|}, \ldots, \frac{a'_m}{\|a'_m\|}$ on $S^{n-1}$ as follows: for each point $\frac{a'_i}{\|a'_i\|}$, its associated cell is defined as
\[
\vor'(i)=\big\{z\in S^{n-1} : \langle a'_i, z \rangle \geq \langle a'_j, z \rangle, \forall j\in [m]\big\}.
\]
Note that $\vor'$ is a partition of $S^{n-1}$; and if we assign a weight $\|a'_i\|$ to each point $\frac{a'_i}{\|a'_i\|}$, then for each point $p\in\vor'(i)$, the site among $\frac{a'_1}{\|a'_1\|},\ldots,\frac{a'_m}{\|a'_m\|}$ that has the smallest {\em weighted} distance to $p$ is $\frac{a'_i}{\|a'_i\|}$.\footnote{The reason of defining such a weighted spherical Voronoi diagram is that we want to have a separating hyperplane between $A$ and $A'=(a'_1,\ldots,a'_m)^T$, rather than $\big(\frac{a'_1}{\|a'_1\|},\ldots,\frac{a'_m}{\|a'_m\|}\big)^T$.}
Note that each cell of $\vor'$ is defined by a set of linear inequalities (other than the unit norm requirement) and each of them can be computed efficiently.

Now we have two diagrams: $\vor$, which is unknown, and $\vor'$, which can be represented efficiently using the $H$-representation. If $\vor\neq \vor'$, then there exists a point $y\in S^{n-1}$ such that $y \in \vor(i)$ and $y \notin \vor'(i)$. Suppose that $y \in \vor'(j)$ for some $j \neq i$. According to the definition, we have $ \langle a_i, y \rangle \geq \langle a_j, y \rangle$ and $\langle a'_i, y \rangle < \langle a'_j, y \rangle$; this gives us a separating hyperplane between $A$ and $A'$. The questions are then (1) how to find such a point $y$ when $\vor \neq \vor'$, and (2) what if $\vor=\vor'$.

\medskip \noindent
\textbf{Consistency check.}
In this part we will show how to check whether $\vor=\vor'$, and if not equal, how to find a $y$ as above. Although we know neither the positions of points $a_1, \ldots, a_m$, nor the corresponding spherical
Voronoi diagram $\vor$, we can still efficiently compare it with $\vor'$, with the help of the oracle.

For each cell $\vor'(i)$, assume that it has $k$ facets (i.e., $(n-1)$-dimensional faces). Note that $k\le m$ and that $\vor'(i)$ is uniquely determined by these facets. Further, each facet is defined by a hyperplane $H'_{ij}=\{z\in S^{n-1}: \langle a'_i, z \rangle = \langle a'_j,z \rangle\}$ for some $j \neq i$. 
To decide whether $\vor=\vor'$, for each $i$ and $j$ such that $\vor'(i) \cap \vor'(j) \neq \emptyset$, we find a sufficiently small $\epsilon_y$ and three points $y$, $y+\epsilon_y$, $y-\epsilon_y$, such that
\[
y\in \vor'(i) \cap \vor'(j)\subset H'_{ij}, \ \ y + \epsilon_y \in \vor'(i)-\vor'(j),\ \ y - \epsilon_y\in \vor'(j)-\vor'(i).
\]
Notice that such $y$ and $\epsilon_y$ exist and can be found efficiently. 
We now query points $y + \epsilon_y$ and $y - \epsilon_y$ to the oracle. If the oracle does return us the expected answers, i.e., $i$ and $j$, respectively, then, with $\|\epsilon_y\|$ sufficiently small (up to $2^{-poly(L)}$), we can conclude that $y$ must also be in the facet of $\vor(i)$ and $\vor(j)$ of the hidden diagram $\vor$. That is, $y \in H_{ij} = \{z\in S^{n-1}: \langle a_i,z \rangle = \langle a_j,z \rangle\}$. We implement the above procedure $n-1$ times to look for $n-1$ linearly independent points $y_1, \ldots, y_{n-1} \in \vor'(i) \cap \vor'(j)$. If the oracle always returns the expected answers $i$ and $j$, respectively, for all $k=1,\ldots,n-1$, then we know that $H_{ij} = H'_{ij}$.

The procedure described above can be implemented in polynomial time. Now we can use this approach to check all facets of all of the cells of $\vor'$. If none of them returns us an unexpected answer, we know that every facet of every cell $\vor'(i)$ is also a facet of cell $\vor(i)$, i.e., the set of linear constraints that defines $\vor'(i)$ is a subset of those that define $\vor(i)$.
Thus, we have $\vor(i) \subseteq \vor'(i)$ for each $i$. Together with the fact that both $\vor$ and $\vor'$ are tessellations of $S^{n-1}$, we can conclude that $\vor = \vor'$.
\begin{Lem}
  For the hidden matrix $A \in \mbR^{mn}$ with spherical Voronoi diagram
  $\vor$ and proposed matrix $A' \in \mbR^{mn}$ with weighted spherical
  Voronoi diagram $\vor'$, we can in polynomial time
  \begin{itemize}
  \addtolength{\itemsep}{-0.5\baselineskip}
  \item either conclude that $\vor = \vor'$, or
  \item find a separating hyperplane between $A$ and $A'$.
  \end{itemize}
\end{Lem}
A formal and detailed description of this consistency check procedure and its correctness proof can be found in the Appendix~\ref{appendix-furthest-formal}.

\medskip \noindent
\textbf{Voronoi diagram recognization.}
If the above process concludes that $\vor=\vor'$, we have successfully found the Voronoi diagram $\vor$ (in its $H$-representation) for the hidden points $a_1,\ldots,a_m$. It was shown by Hartvigsen~\cite{Hart92} that given a Voronoi diagram with its $H$-representation, 
a set of points that generates the diagram can be computed efficiently. Further, Ash and Bolker~\cite{AB85} showed that the set of points that generates a non-degenerate Voronoi diagram is unique. Therefore, by coupling these two results and the assumption that the input matrix $A$ is non-degenerate, we are able to identify the positions of $a_1,\ldots,a_m$ given the computed Voronoi diagram $\vor$, and easily determine if the LP $Ax>0$ has a feasible solution, and compute one if it exists.

\subsection{The General {\large $Ax>b$}}\label{sec:Ax>b}

In this section, we extend our algorithm to the general case $Ax > b$. Due to space limit, we will only give the main ideas in this section and leave the formal proof to Appendix~\ref{appendix-furthest-formal}.

The idea is still to use the ellipsoid method to find the unknown
point $(A, b)$, which is a degenerate polyhedron in $\mbR^{m(n+1)}$. Our goal is, for any
considered point $(A', b') \in \mbR^{m(n+1)}$ of the center of the ellipsoid, to compute a hyperplane that separates it from the true point
$(A, b)$. Again, if $x$ is a feasible solution of $A'x > b'$, we can
query it to the oracle and using the returned index to get a separating hyperplane. Thus, in the following,
we assume that the LP $A'x > b'$ is infeasible.

\medskip \noindent
\textbf{Generalized furthest Voronoi diagram.}
Similar to the previous case, for the hidden LP $Ax > b$,
we assume without loss of generality that $\|a_i\| = 1$, for all $1\le i\le m$.
We denote by $\gvor$ the tessellation of $\mbR^n$ into polyhedra $\gvor(1), \ldots,
\gvor(m)$, where
\begin{equation}\label{eq:GV}
	\gvor(i) = \big\{x \in \mbR^n : b_i - \langle a_i, x \rangle \geq b_j - \langle a_j, x \rangle, \ j\in [m]\big\}.
\end{equation}
Note that different from the $Ax>0$ case, $\gvor$ is no longer a spherical closest Voronoi diagram; it can be seen as a generalized {\em furthest} Voronoi diagram over $\mbR^n$ where each source site is a half-space. It follows that for any query $x \in \mbR^n$, if $x$ is not a feasible solution to the LP $Ax > b$, the oracle always returns an index $i$ where $x \in \gvor(i)$. 

For any presumed point $(A', b') \in \mbR^{m(n+1)}$, we compute
\begin{equation}\label{eq:GV'}
	\gvor'(i) =\big\{x \in \mbR^n : b'_i - \langle a'_i, x \rangle \geq b'_j - \langle a'_j, x \rangle, \ j\in [m]\big\}
\end{equation}
and denote $\gvor' = \big\{\gvor'(1), \ldots, \gvor'(m)\big\}$. Similar to the
previous case, if $\gvor \neq \gvor'$, we want to find a point $x \in
\mbR^n$ such that $x \in \gvor(i)$ and $x \in \gvor'(j)$ for some $i
\neq j$; this gives us $b_i - \langle a_i, x \rangle \geq b_j -
\langle a_j, x \rangle$ and $b'_i - \langle a'_i, x \rangle \leq b'_j
- \langle a'_j, x \rangle$, with at least one of the inequalities
strict. Thus, we have a separating hyperplane between $(A, b)$ and
$(A', b')$. Again because each cell $\gvor'(i)$ (as well as
$\gvor(i)$) is a polytope defined by a set of linear inequalities, and $\gvor'$ and $\gvor$ are a tessellation of the space $\mbR^n$,
we can apply the same procedure as the $Ax > 0$ case to check whether $\gvor = \gvor'$ in polynomial time (even if $\gvor'(i)$ or $\gvor(i)$ is empty).

However, even after finding the tessellation $\gvor$ for the hidden LP $Ax
> b$, i.e., $\gvor = \gvor'$, we still cannot recover the actual instance $A$
and $b$ from $\gvor$ even all half-spaces are non-degenerate. For instance,
for any scale $c \in \mbR$, the two systems $Ax > b$ and $Ax > b + \{c, c, \ldots,
c\}^T$ have the exactly same tessellation $\gvor$, but they may
have different feasible solutions. In the following, we will show that when
$\gvor = \gvor'$,
we are able to focus on a particular point in $\mbR^n$, and use the claims proved in the previous section to solve the problem.

\medskip \noindent
\textbf{Extreme point.}
Given a matrix $A$ and a vector $b$, define
\begin{equation}\label{eq:vertexdist}
	d(A, b) = \min_{x\in \mathbb{R}^n}\big\{\max_i\{b_i - \langle a_i, x \rangle\}\big\}, \mbox{ and}
\end{equation}
\begin{equation}\label{eq:extreme}
	\extreme(A, b) = \big\{x \in \mathbb{R}^n: \max_i\{b_i - \langle a_i, x \rangle\} = d(A,b)\big\}.
\end{equation}
Notice that $d(A, b)$ may be unbounded for general $A$ and $b$, e.g., when $Ax>b$ has an unbounded feasible region. For such a case, we define $d(A,b)$ to be $-\infty$ and $\extreme(A,b)=\emptyset$. But when $Ax > b$ is infeasible, $\extreme(A, b)$ is always nonempty and $d(A,b)$ is bounded by 0 from below.
The definition of $\extreme(A, b)$ gives the set of points that minimizes the maximal distance from which to all half-spaces $\{z\in \mathbb{R}^n : \langle a_i, z \rangle \ge b_i\}$. Note that given matrix $A$ and vector $b$, both $d(A, b)$ and a point in $\extreme(A,b)$ can be computed efficiently through linear programming. (The value $d(A, b)$ is an LP and $\extreme(A,b)$ is the union of the feasible regions of $m$ LP's.)


The next lemma links the homogeneous and non-homogeneous forms and will be later used to decide the infeasibility.
\begin{Lem} \label{lem:ax>b}\label{lemma-furthest}
  A linear system $Ax>b$ is infeasible if and only if there is a point $p\in \mathbb{R}^n$ such that
  \begin{itemize}[$\bullet$]
  \addtolength{\itemsep}{-0.5\baselineskip}
  \item $p$ is not a feasible solution of $Ax>b$, and
  \item the linear system $\{\langle a_i, x \rangle > 0 : i \in S\}$, called the {\em support linear system} of $Ax>b$ at $p$, is infeasible, where $S$ is the set of the indices of the half-spaces in the LP \textbf{$Ax>b$} that has the (same) maximal distance to point $p$.
  \end{itemize}
  Further, if $Ax>b$ is infeasible, then any point in $\extreme(A,b)$ satisfies the above two conditions.
\end{Lem}

The proof of the lemma is deferred to Appendix~\ref{appendix-lemma-furthest}.
We will use the lemma to infer that the hidden linear system $Ax > b$ is infeasible.
For the considered LP $A'x>b'$, as it is infeasible, $\extreme(A', b')\neq \emptyset$; we then compute an extreme point $p\in \extreme(A', b')$.

Now, if we focus on the region around $p$ and limit our queries within the ball $B = \{z\in \mathbb{R}^n : \|z - p\| \leq \epsilon^*\}$ for some small enough $\epsilon^* > 0$, then for each query $x = p + x'$ within the ball $B$, where $\|x'\|\le \epsilon^*$, the oracle returns an index
\begin{eqnarray*}
i &\in & \arg\max_{i\in [m]}\big\{b_i - \langle a_i, x \rangle\big\} = \arg\max_{i\in S}\big\{b_i - \langle a_i, x \rangle\big\} \\
  &=& \arg\max_{i \in S}\big\{b_i - \langle a_i, p \rangle -\langle a_i, x' \rangle\big\} = \arg\max_{i \in S}\big\{-\langle a_i, x' \rangle\big\},
\end{eqnarray*}
where the first equality is because $\epsilon^*$ is sufficiently small, and the last equality follows from the fact that all half-spaces in $S$ have the same maximal distance to $p$.
Thus, for any queried point within this ball $B$, the set of indices possibly returned by the oracle for $Ax > b$ is the same as that by the furthest oracle for the support linear system $\{\langle a_i, x \rangle > 0 : i\in S\}$.
This means that the Voronoi structure $\gvor$ in $B$ is exactly the same as the weighted spherical Voronoi diagram $\gvor'$ for the corresponding support system in $B$. Then, by the results in previous section, we are able to identify the exact values for all $a_i$ in $S$; that is, the support system $\{\langle a_i, x \rangle > 0 : i\in S\}$ is revealed. 
The last step is straightforward: if the support system is infeasible, then by Lemma~\ref{lem:ax>b}, we can conclude that LP $Ax > b$ is also infeasible. If the support system has a feasible solution $x^*$, since we know $x^*$ is not a feasible solution to $\{\langle a'_i, x \rangle > 0 : i\in S\}$ (by Lemma~\ref{lem:ax>b}), we must have $\langle a'_i, x^* \rangle \leq 0 < \langle a_i, x^* \rangle$ for some $i \in S$, which is a separating hyperplane between $(A, b)$ and $(A', b')$.


\medskip \noindent
\textbf{Putting things together.} To summarize the above discussions, in our algorithm we employ the ellipsoid method to search for $(A,b)\in \mbR^{m(n+1)}$. In each iteration of the ellipsoid method, we propose the center
$(A', b') \in \mbR^{m(n+1)}$ of the current ellipsoid, and apply the following procedure:

\begin{enumerate}[1.]
\addtolength{\itemsep}{-0.5\baselineskip}
\item If $A'x > b'$ is feasible, find a solution $x$ to it and query $x$ to the oracle, then use the returned index to construct a separating hyperplane.
\item Otherwise (i.e., $A'x>b'$ is infeasible), do the following.
\begin{itemize}
\addtolength{\itemsep}{-0.4\baselineskip}
\item Compute the generalized Voronoi diagram $\gvor'$ of $A'$ and $b'$, and confirm that $\gvor = \gvor'$.
\item Compute a point $p\in \extreme(A', b')$ and the corresponding support $S$, then confirm that $p$ is an infeasible solution to $Ax > b$.
\item Focus on a small ball $B$ centered at $p$ and use $\gvor = \gvor'$ in $B$ to recover the support linear system at $p$. Then
  confirm that the support system is infeasible, and, by Lemma~\ref{lem:ax>b}, conclude that the hidden LP $Ax > b$ is also infeasible and terminate the whole program.
\end{itemize}
\end{enumerate}

If we get an ``unexpected'' answer from the oracle at any step of the above procedure, then we either receive a \yes from the oracle and thus solve the problem, or get a hyperplane that separates the unknown $(A,b)$ and current center $(A', b')$, in which case we jump out of the current iteration and continue with the ellipsoid method with a smaller ellipsoid. From the above discussions, we know that every step can be implemented in polynomial time. Hence, the problem can be solved efficiently.

A complete description of the algorithm and its formal proof can be found in Appendix~\ref{appendix-furthest-formal}.


\section{Worst-Case Oracle}\label{section-worst}

In this section, we consider the worst-case oracle. Recall that in this setting, the oracle plays as an
adversary by giving the worst-case violation index to force an algorithm to use the maximum amount of time to solve
the problem.

For any linear program $Ax > b$, we can introduce another variable $y$ and transform the linear program into the
following form:
\begin{eqnarray*}
  Ax - by > 0 \\
  y > 0
\end{eqnarray*}
It is easy to check that $Ax > b$ is feasible if and only if the new LP is feasible, and the solutions of these two linear systems can be easily transformed to each other. Given the oracle for $Ax>b$, one can also get another oracle for the new LP easily. (On a query $(x,y)$, if $y \leq 0$, return the index $m+1$; otherwise, query $x/y$ to the oracle for $Ax>b$.) This means that the \ulp problem of the homogeneous form $Ax > 0$ is no easier than the problem of the general form. In all the analysis of this section, we will therefore only consider the problem of form $Ax > 0$.

\medskip \noindent
\textbf{Geometric explanations.}
Let us consider the problem from a geometric viewpoint. Any matrix $A=(a_{ij})_{m\times n}$
can be considered as $m$ points $\ap_1, \ap_2, \ldots, \ap_m$ in the $n$-dimensional space $\mathbb{R}^n$,
where each $\ap_i = (a_{i1}, a_{i2}, \ldots, a_{in})$. The positions of these points are unknown to us.
Finding a feasible solution $x\in \mathbb{R}^n$ that satisfies $Ax > 0$ is equivalent to finding an open half-space
\[
H_x = \big\{y\in \mbR^n :\  \langle x, y \rangle \triangleq x_1y_1 + x_2y_2 + \cdots + x_ny_n > 0\big\}
\]
containing all points $a_i$. 


In an algorithm, we propose a sequence of candidate solutions. When a query $x\in \mathbb{R}^n$ violates a constraint $i$, we know that $\langle \ap_i, x \rangle \le 0$. Hence, $\ap_i$ cannot be contained in the half-space $H_x$, and we are able to cut $H_x$ off from the possible region of $\ap_i$. Based on this observation, we maintain a set $\region(i)$, the region of
possible positions of point $\ap_i$ consistent with the information obtained from the previous queries. Initially, no information is known about the position of any point; thus, $\region(i) = \mathbb{R}^n$ for all $1 \leq i \leq m$.

Let us have a closer look at these regions. For each $i$, suppose that $x^i_1, x^i_2, \ldots, x^i_k$ are the queried points we have made so far for which the oracle returns index $i$. Then all information we know about $\ap_i$ till this point
is that the possible region is $\region(i) = \bigcap_{j=1}^k \{y\in \mbR^n : \langle x_j^i, y \rangle \leq 0\}$. Since $\region(i)$ is the intersection of $k$ closed half-spaces, it is a convex set. Equivalently, this means that any feasible solution to the LP, if existing, cannot be in $\region(i)^*$, the polar cone of $\region(i)$. Since the polar cone of a half-space
$\{y\in \mathbb{R}^n \mid \langle x^i_j, y \rangle \leq 0\}$ is the ray along its normal vector, i.e., $\{\lambda x^i_j \mid \lambda \geq 0\}$,  we have by Lemma~\ref{lem:polar} that
\begin{eqnarray*}\small
\region(i)^* = conv\Big(\bigcup_j{\big\{y \mid \langle x^i_j, y \rangle \leq 0\big\}^*}\Big) = conv\Big(\big\{\lambda x^i_j \mid 1 \leq j \leq k, \lambda \geq 0\big\}\Big).
\end{eqnarray*}
\normalsize
Since $\region(i)^*$'s are the forbidden areas for any feasible solution, we can conclude that the LP has no feasible solution if $\bigcup_i\region(i)^* = \mathbb{R}^n$.



\medskip \noindent
\textbf{Convex hull covering algorithms.}
Based on above observations, we now sketch a framework of \emph{convex hull covering algorithms} that solves the \ulp problem. The algorithm maintains a list of $m$ convex cones
$$\region(1)^*, \region(2)^*, \ldots, \region(m)^* \subseteq \mathbb{R}^n.$$
Initially, $\region(i)^* = \emptyset$ for all $1 \leq i \leq m$. On each query $x \in \mathbb{R}^n$, the oracle either returns \yes, indicating that the problem is solved, or returns us an index $i$, in which case we update $\region(i)^*$ to $conv\left(\region(i)^*, \{\lambda x \mid \lambda > 0\}\right)$.
The algorithm terminates when either the oracle returns \yes, or when $\mathbb{R}^n - \bigcup_i\region(i)^*$ does not contain a convex cone with normalized volume at least $2^{-(2n+3)L}$, which indicates that the given instance has no feasible solution. The above discussion can be formalized into the following theorem.
\begin{Thm}\label{thm:chca}
  Any algorithm that falls into the convex hull covering algorithm framework solves the \ulp problem.
\end{Thm}

Though the framework guarantees the correctness, it does not specify how to make queries to control complexity. Next we will show an algorithm with nearly optimal complexity.

\subsection{Warmup: 2-Dimension}

To illustrate the idea of our algorithm, we consider the simplest case in which the number of variables is 2. In a 2-dimensional plane, using the polar coordinate system, every closed convex cone can be represented as an interval
of angles $[\alpha, \beta]$ where $0 \leq \alpha, \beta \le 2\pi$ (e.g.,
$[0,\pi/2]$ represents the first quadrant).

Our algorithm sets $\region(1)^* = \cdots = \region(m)^* = \emptyset$ initially. At each step of the algorithm, assume that $[0, 2\pi] - \bigcup_i\region(i)^* = \bigcup_{i=1}^k (\alpha_i, \beta_i)$, where $(\alpha_i, \beta_i) \cap (\alpha_j, \beta_j) = \emptyset$ for any $1 \leq i < j \leq k$. We pick an interval $(\alpha_t, \beta_t)$ with the maximum $\beta_t-\alpha_t$, and query the oracle on point $(\cos(\gamma), \sin(\gamma))$ where $\gamma = (\alpha_t+\beta_t)/2$.
Suppose the oracle returns an index $i$, we then update $\region(i)^*$ to $conv(\region(i)^*, \gamma)$, where $conv(S)$ here is the smallest sector containing all points in $S$.
The algorithm iteratively runs the above procedure, until at some point we have $(\beta_t - \alpha_t)/\pi < 2^{-(2n+3)L}$, when we can conclude that the LP has no feasible solution.

%
%
%

\begin{figure}[ht]
\begin{center}
\includegraphics[scale = 0.8]{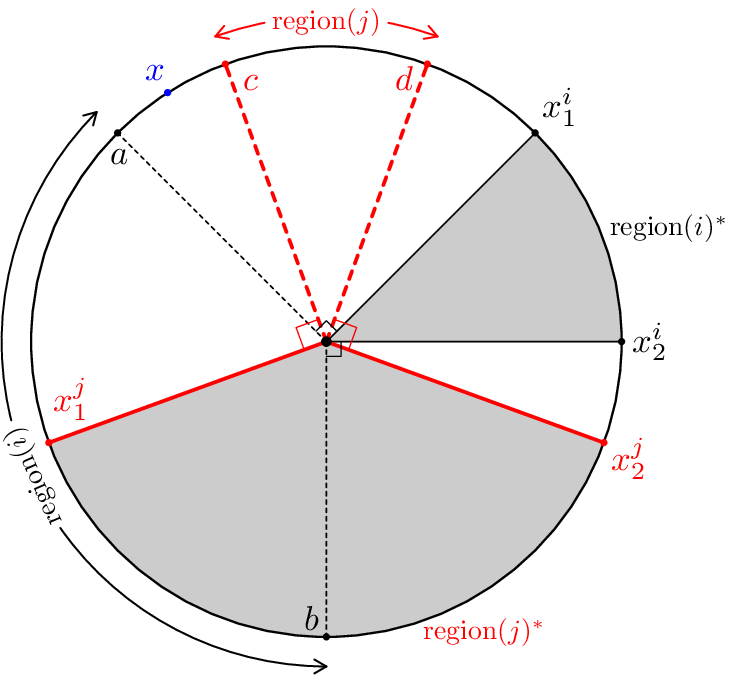}
\end{center}
\end{figure}

Consider the above figure for an example. Let us say that when we query $x_1^i$ and $x_2^i$, the oracle returns index $i$. Then the $i$-th constraint can only be in $[a,b]$, where $a$ and $b$ are perpendicular to $x_1^i$ and $x_2^i$, respectively, and no feasible solution of the LP can be in $[x_2^i,x_1^i]$. That is, $\region(i)=[a,b]$ and its polar cone $\region(i)^*=[x_2^i,x_1^i]$,
which is the convex hull of $x_1^i$ and $x_2^i$. Similarly, if $x_1^j$ and $x_2^j$ are two queries in which the oracle returns another index $j$, we have $\region(j)^*=[x_1^j,x_2^j]$ and $\region(j)=[c,d]$.
Now given obtained information, we have
$[0,2\pi]-\bigcup_i\region(i)^* = (x_1^i,x_1^j)\cup (x_2^j,x_2^i)$. Among the two, $(x_1^i,x_1^j)$ has the maximum length and will be picked by the algorithm.
Then the middle point ($x$ in the figure) of the interval $(x_1^i,x_1^j)$ will be picked and queried to the oracle. If the oracle returns index $i$,
then $\region(i)^*$ becomes $[x_2^i,x]$; in which case the length of the interval $(x_1^i,x_1^j)$ is cut into half. In this example, the oracle cannot return index $j$ since the entire candidate region $\region(j)$ has positive inner product with $x$.
If the oracle returns a new index $k$, then we have $\region(k)^*=[x,x]$, and in the next iteration of the algorithm, we have $[0,2\pi]-\bigcup_i\region(i)^* = (x_1^i,x)\cup (x,x_1^j)\cup (x_2^j,x_2^i)$.

Note that first, the number of candidate intervals is at most $n$. Second, each query either increases the number of intervals or cut the length of a current interval into half. As there is a lower bound on the length of each interval in the algorithm (Lemma~\ref{lem:feasiblesize}), the algorithm terminates with a feasible solution or claims that no feasible solution exists.

\subsection{Algorithm}

In this section, we generalize our algorithm from 2-dimensional to $n$-dimensional. The basic idea is to use induction on dimension. That is, we pick an $(n-1)$-dimensional subspace and recursively solve the problem on the subspace. The subroutine either finds a point $x$ in the subspace that satisfies $Ax>0$ (in which case the algorithm ends), or finds out that there is no feasible solution in the entire subspace. In the latter case, the whole space of candidate solutions can be divided into two open half-spaces, and we will work on each of them separately. In general, we have a collection of connected regions that can still contain a valid solution. These regions are the ``holes'', formally called \emph{chambers}, separated by $\bigcup_i \region(i)^*$ (recall that points in $\region(i)^*$ cannot be a feasible solution). We can then pick a chamber with the largest volume, and cut it into two balanced halves by calling the subroutine on the hyperplane slicing the chamber.

There are several issues for the above approach. The main one is that there may be too many chambers: \emph{a priori}, the number can grow exponentially with $m$. There are also other technical issues to be handled, such as how to represent chambers (which are generally concave), how to compute (even approximately) the volume of chambers, how to find a hyperplane to cut a chamber into two balanced halves, etc.

For the first and main issue, it can be shown that the number of chambers cannot be too large. In general, as the later Theorem~\ref{thm:chambers} shows, any $m$ convex sets in $\mbR^n$ cannot form more than $\sum_{i=1}^n\binom{m}{i}$ chambers.
For the rest technical issues, we deal with them in the following way. Instead of keeping track of all actual chambers, in our algorithm, we maintain a collection of disjoint \emph{sector cylinders}, which can be shown to be supersets of chambers. This greatly simplifies the main algorithm (i.e., the induction part) to a procedure which is very similar to the 2-dimensional case as described in the previous section. Furthermore, we only keep those cylinders that contain at least one chamber, thus, the bound for the number of chambers also bounds the number of cylinders from above.



The algorithm is formally given as below. We call the program \AlgUnknownLP$(\mathbb{R}^n)$ to get a solution of the \ulp problem. Note that each subroutine has its own local variables, and all subroutines share the same global variables.

\begin{algorithm}[H]
	\caption{\AlgUnknownLP$(V)$} 
  {\bf Input}:
	$V$: A subspace of $\mbR^n$. \\
	{\bf Output}: A feasible solution $x\in V$, or \no (solution in $V$). \\[-0.12in]
	
	Global variables: \mbox{$\region(1)^*, \ldots, \region(m)^*$ (initially all $\emptyset$).} \\
	Local variables: \Cylinders. \\[-0.12in]
	
  \begin{algorithmic}[1]
    \STATE Let $d$ be the dimension of $V$, and $\{b_1, b_2, \ldots, b_d\}$ be an orthonormal basis of $V$.
		\IF {$d = 1$}
				\STATE Make two queries $x = b_1$ and $x = -b_1$ to the oracle. \label{alg:propose}
				\IF {the oracle returns \yes on an $x$}
					\RETURN $x$ and halt the whole program \AlgUnknownLP$(\mbR^n)$.
				\ENDIF
				\STATE Update $\region(i)^*=conv\big(\region(i)^*,\{\lambda x : \lambda\ge 0\}\big)$ for each returned $i$ in the above queries. \label{alg:update}
                                \STATE Halt the current program \AlgUnknownLP$(V)$.
		\ENDIF
		\STATE Run \AlgUnknownLP$\big(\big\{x = \sum_i x_ib_i \in V: \atan(x_1, x_2) = 0\big\}\big)$. \label{alg:induction}
		\STATE Let $\Cylinders = \{(0,2\pi)\}$.
    \WHILE {\TRUE}
      \STATE Pick $(\alpha,\beta) \in \Cylinders$ with the max $\beta - \alpha$. 
      \IF {$(\beta-\alpha)/\pi < 2^{-(2n+3)L}$} \label{alg:no}
				\RETURN \no and terminate the current program \AlgUnknownLP$(V)$.
      \ELSE
        \STATE Let $\gamma = (\alpha + \beta) / 2$ and \mbox{run \AlgUnknownLP$\big(\big\{x = \sum_i x_ib_i \in V:\ \atan(x_1, x_2) = \gamma\big\}\big)$.} \\[0.02in] \label{alg:subspace}
		\STATE Replace $(\alpha,\beta)$ in $\Cylinders$ by $(\alpha, \gamma)$ and $(\gamma, \beta)$. \label{alg:half}
					\FOR {each element $(\delta,\lambda)$ in $\Cylinders$}
						\IF {$\big\{x = \sum_i x_ib_i \in V : \atan(x_1, x_2)\in (\delta,\lambda)\big\} \subseteq \bigcup_i\region(i)^*$} \label{alg:contain}
							\STATE Remove $(\delta,\lambda)$ from $\Cylinders$ \label{alg:remove}
						\ENDIF	
					\ENDFOR
      \ENDIF
    \ENDWHILE
  \end{algorithmic}
\end{algorithm}

In the algorithm, $\atan$ function is a two-argument variant of the arctangent function: $\atan(x, y)$ is
the counter-clockwise angle between the positive horizonal axis and the point $(x, y)$ on the plane.
The induction step is in line~\ref{alg:induction} and~\ref{alg:subspace}, where the subspaces called are defined by $\atan(\cdot)$ of the first two coordinates. The algorithm either returns a feasible solution $x$ in $V$, or discovers that the entire subspace of $V$ does not contain any feasible solution by keeping track of $\region(i)^*$. The algorithm thus falls into the framework of convex hull covering algorithms and always returns a correct answer.

To implement induction in the algorithm, the information of an orthonormal basis of a called subspace can be derived from the current considered space.
In particular, the subspace $\big\{x = \sum_i x_ib_i \in V: \atan(x_1, x_2) = 0\big\}$ in line~\ref{alg:induction} is equivalent to
the space expanded by the basis $(\cos(0)b_1 + \sin(0)b_2, b_3, \ldots, b_d)$, and
the subspace given by $\atan(x_1, x_2) = \gamma$ in line~\ref{alg:subspace} is equivalent to the space expanded by $(\cos(\gamma)b_1 + \sin(\gamma)b_2, b_3, \ldots, b_d)$.
Thus, the subspace with parameter $V$ in the algorithm can be implemented efficiently.

Within the program with parameter $V$, each element $(\alpha,\beta)\in \Cylinders$ corresponds to
$\big\{x = \sum_i x_ib_i\in V: \atan(x_1,x_2)\in (\alpha,\beta)\big\},$
which is a sector cylinder in subspace $V$. It can be shown that a cylinder surviving at the end of each while-loop iteration
must contain at least one entire chamber; thus, the upper bound for the number of chambers, which is at most $O(m^{n})$, also bounds the number of cylinders.
The volume of a maximum cylinder is cut by half (line \ref{alg:half}) and a cylinder is disqualified if it is smaller than $2^{-(2n+3)L}$ (line \ref{alg:no}), the algorithm takes at most $m^{poly(n)}(2n+3)L$ iterations in the while loop. Taking the recursion into consideration, there are at most $O\big(m^{poly(n)} poly(n)\cdot L)\big)$ iterations executed.



The above analysis leads to the following theorem.

\begin{Thm} \label{thm:general}
  \AlgUnknownLP$(\mbR^n)$ solves the \ulp problem in time $O\big((mnL)^{poly(n)}\big)$. In particular, the problem can be solved in polynomial time if the LP has constant variables.
\end{Thm}

\begin{proof}
We prove the following property by induction: when called on a parameter $V$, the program either returns a feasible solution $x$ in $V$, or discovers that the entire subspace $V$ does not contain any feasible solution. Once this is proved, applying it to the case $dim(V)=1$ gives the correctness of the algorithm.
The induction base is trivially true.
Suppose that the claim is true for $dim(V) = d-1$, and we consider the case for $dim(V) = d$. It is not hard to see that the algorithm falls into the framework of convex hull covering algorithms: Starting from $\region(i)^* = \emptyset$, the algorithm proposes queries (line \ref{alg:propose}) and uses the returned violations to update $\region(i)^*$'s (line \ref{alg:update}).
Thus, Theorem~\ref{thm:chca} guarantees the correctness of the algorithm, as long as the algorithm always terminates. We will show this together with the complexity analysis below.

Within the program with parameter $V$ with an orthonormal basis $\{b_1, b_2, \ldots, b_d\}$, each element $(\alpha,\beta)\in \Cylinders$ corresponds the set
$$\bigg\{x = \sum_i b_ix_i\in V ~\Big|~ \atan(x_1,x_2)\in (\alpha,\beta)\bigg\},$$
which is a sector cylinder in subspace $V$. Each cylinder surviving at the end of each while-loop iteration intersects with some chamber, because otherwise it is in $\bigcup_i$ $\region(i)^*$ and thus has been eliminated in line \ref{alg:remove}. Note that the boundary of the sector cylinder $(\alpha,\beta)$ are the two $(d-1)$-dimensional subspaces
$$V_\alpha = \bigg\{x = \sum_ib_ix_i \in V ~\Big|~ \atan(x_1, x_2) = \alpha\bigg\}$$ and
$$V_\beta = \bigg\{x = \sum_ib_ix_i \in V ~\Big|~ \atan(x_1, x_2) = \beta\bigg\},$$
both of which have been searched in the previous loop iterations for feasible solutions. By induction hypothesis, either a feasible solution has been found and the whole program has ended, or we have known at this point that the two subspaces do not contain a feasible solution. Since a chamber intersects with the cylinder, but not with the cylinder's boundary, we can conclude that the cylinder contains at least one entire chamber. Therefore, when all cylinders are too small to contain any feasible region (line \ref{alg:no}), we can conclude that the current subspace $V$ does not contain any feasible solution. To be more precise, first, note that the normalized volume of a sector cylinder
\[
\{x\in \mbR^n: \|x\| = 1, \atan(x_1,x_2)\in (\alpha, \beta)\}
\]
is the same as that of a 2-dimensional sector
\[
\{x\in \mbR^2: \|x\| = 1, \atan(x_1,x_2)\in (\alpha, \beta)\},
\]
both equal to $(\beta - \alpha)/2\pi$. Second, in general it is not true that a set $S$ of small volume cannot contain a solution. But the set in our case is the sector cylinder, and the induction hypothesis guarantees that the boundary (the two half-subspaces corresponding to $\atan(x_1,x_2) = \alpha$ and $\beta$) do not contain any feasible solution. Thus the cylinder either contain no feasible solution, or the entire convex cone of the feasible region in Lemma \ref{lem:feasiblesize}.

Since each cylinder contains at least one entire chamber, and different cylinders clearly do not intersect, the upper bound for the number of chambers also holds for the number of cylinders. As each time the volume of the maximum cylinder shrinks by half (line \ref{alg:half}) until it is smaller than $2^{-(2n+3)L}$ (line \ref{alg:no}), by Theorem~\ref{thm:chambers}, the algorithm takes at most $\sum_{i=0}^{d}\binom{m}{i} (2n+3)L$ iterations in the while loop. Taking the recursion into consideration, there are at most
$$\prod_{d=0}^n \ \left(\sum_{i=0}^{d}\binom{m}{i} (2n+3)L \right) = O\big(m^{n^2}n^nL^n\big)$$
iterations executed.

Inside each loop, the only steps other than the subroutines that cost
more than a constant amount of time is at line \ref{alg:contain},
which is to check whether $\bigcup_i \region(i)^*$ contains some set. Here we briefly argue how to do it in $O\big(m^{n^3}\big)$ time. Each $\region(i)^*$ has at most $m^n$ facets and thus $\bigcup_i \region(i)^*$ has at most $\binom{m^n}{n} \leq m^{n^2}$ vertices. These vertices can form at most $\binom{m^{n^2}}{n+1}=O\big(m^{n^3}\big)$ simplexes, which are convex. For each of these potential simplexes, pick an arbitrary interior point $p$ and check whether it is outside $\bigcup_i \region(i)^*$. 

Therefore, the whole algorithm \AlgUnknownLP$(\mathbb{R}^n)$ has running time $O\big(m^{n^3}n^nL^n\big)$.
\end{proof}

\subsection{Counting the Number of Chambers}

Consider the union of $m$ polytopes in $\mathbb{R}^n$, their complement divides the whole space $\mathbb{R}^n$
into a number of disconnected components, called {\em chambers}. To analyze the running time of our algorithm,
we need to count the number of chambers formed by $m$ polytopes (or more generally, convex sets).
This question was first raised by L{\' a}szl{\' o} Fejes T{\' o}th as an open problem. The 2-dimensional case was proved by Katona~\cite{Kat77} and the general case was proved by Kovalev~\cite{Kov88}.

\begin{Thm}[Kovalev~\cite{Kov88}]\label{thm:chambers}
The complement of the union of $m$ open or closed convex sets in $\mathbb{R}^n$ can have at most $\sum_{i = 0}^n{m \choose i}=O(m^{n})$ chambers.
\end{Thm}

The bound described in the theorem is tight: consider, e.g., when all convex sets are hyperplanes.
The proof of Katona was based on the analysis of the shape of the convex sets, and the proof of Kovalev used induction on the dimension. 
Next, we give an alternative and simpler proof to this theorem. Our proof does not rely on induction and is independent to~\cite{Kat77,Kov88}. From our proof, we can see clearly where each binomial coefficient in the summation comes from.
(We recommend readers read all these proofs for comparison.)

We will use bounded polygons in a 2-dimensional plane $\mathbb{R}^2$ to illustrate the idea of our proof,
which is completely elementary and much simpler than the one in \cite{Kat77}. 
Actually, we identify the measure of exterior angles\footnote{An \emph{exterior angle} of a polygon at a vertex $v$ is defined as the angle formed by one side adjacent to $v$ and a line extended from the other side. The value is $\pi$ minus the interior angle. For a concave vertex, its exterior angle is negative.}
to naturally bridge the number of polygons and that of chambers, and we only use the well-known exterior angle theorem which says that the sum of exterior angles of a 2-dimensional polygon is $2\pi$.
For a given polygon $C$, let $V(C)$ denotes the set of vertices of $C$ and let $\alpha(v, C)$ denote the exterior angle at vertex $v$,
then $\sum_{v \in V(C)}\alpha(v, C) = 2\pi$.

Consider $m$ bounded and closed convex polygons $C_1, C_2, \ldots,$ $C_m \subseteq \mbR^2$. For simplicity, we assume that no three edges intersect at the same point. Assume that the complement of their union, $\mbR^2 \backslash \bigcup_iC_i$, has $k+1$ chambers $D_0, D_1, \ldots, D_k$, where $D_0$ is unbounded and $D_1, \ldots, D_k$ are bounded polygons.
Then for any vertex $v \in V(D_j)$, there are two possibilities (see Figure \ref{figure-chamber-2d} for an illustration):
\begin{itemize}[$\bullet$]
\addtolength{\itemsep}{-0.5\baselineskip}
\item $v$ is a vertex of some polygon $C_i$. In this case we have $\alpha(v, D_j) < 0 < \alpha(v, C_i)$.
\item $v$ is the intersection point of two edges from two polygons $C_i$ and $C_{i'}$. In this case we have $\alpha(v, D_j) = \alpha(v, C_i \cap C_{i'})$.
\end{itemize}

\begin{figure}[ht]
\begin{center}
\includegraphics[scale = 0.7]{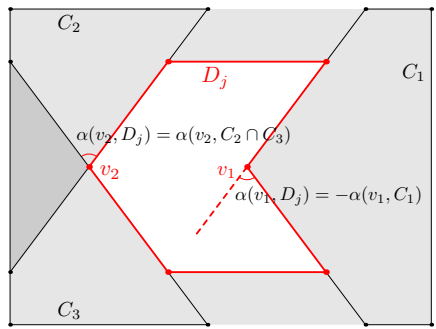}
\caption{Illustration of the 2D proof.}\label{figure-chamber-2d}
\end{center}
\end{figure}

Hence, all exterior angles of $D_j$ are ``contributed" by the $m$ polygons and their intersections.
Adding up all these exterior angles, we have\footnote{The second inequality only considers the intersection of two but not more $C_i$'s, because we assumed that no three edges intersect at the same point.
When there are such three edges, it is actually to the advantage of our analysis.} 
\small
\begin{eqnarray*}
  2\pi \cdot k & = & \sum_{j=1}^k\sum_{v \in V(D_j)}\alpha(v, D_j) \leq \sum_{j=1}^k\sum_{v \in V(D_j)}\alpha\bigg(v, \bigcap_{v \in C_i}C_i\bigg) \\
  & \leq & \sum_{i=1}^m\sum_{v \in V(C_i)}\alpha(v, C_i) + \sum_{1 \leq i < j \leq m}\sum_{v \in V(C_i \cap C_j)}\alpha(v, C_i \cap C_j) = 2\pi \cdot \left({m \choose 1} + {m \choose 2}\right)
\end{eqnarray*}
\normalsize
Thus, the total number of chambers is at most $\sum_{i=0}^2{m \choose i}$.

In order to generalize the proof to general convex sets in the $n$-dimensional space $\mathbb{R}^n$, one needs to
define exterior angles for general convex sets in higher dimensions, so that it can still capture the intrinsic relationship between convex sets and chambers. To this end, we adapt the concept of ``extreme directions'' from Banchoff's seminal work~\cite{Ban67} and define the exterior angle of a boundary point as the set of all extreme directions at that point. Our definition of exterior angles is different from that of in~\cite{Ban67}, in that we drop all low dimensional terms. The reason is that Gauss-Bonnet-type theorems relate total curvature to Euler characteristic, but for our purpose, the invariant of total curvature (i.e., exterior angle in our definition) in the highest dimension precisely links the convex sets and chambers. The formal proof is deferred to Appendix~\ref{appendix-chamber}.


\subsection{Lower Bound}

In this section, we establish an exponential lower bound for the \ulp problem. We will need McMullen's celebrated proof ~\cite{McM70,MS71} of the seminal Upper Bound Conjecture by cyclic polytopes.

Consider the moment curve $c: \mathbb{R} \rightarrow \mathbb{R}^n$ that defines
$c(t) = (t, t^2, \ldots, t^n)$ for $t \in \mathbb{R}$.
For any distinct $m>n$ points $c(t_1),\ldots,c(t_m)$ on the moment curve,
its convex hull $conv(c(t_1),\ldots,c(t_m))$ is called a {\it cyclic polytope} with $m$ vertices.
It is known that its combinatorial structure (including the number of faces of any dimension) is uniquely determined by $n$ and $m$, and is independent of the points chosen.
We use $C(n, m)$ to denote such a {\it cyclic polytope} with $m$ vertices in the $n$-dimensional space.

\begin{Thm}[Upper Bound Theorem (in dual form)]\label{thm:upperb}
  Let $f_k(P)$ denote the number of $k$-dimensional faces of a polytope $P$.
  For any polytope $P$ in $\mathbb{R}^n$ with $m$ facets (i.e., $(n-1)$-dimensional faces) and any $0
  \leq k \leq n-2$, we have $f_k(P) \leq f_{n-k-1}(C(n, m)).$
  In particular,
 $$f_0(P) \leq f_{n-1}(C(n, m)) = {{m - \lfloor \frac{n+1}{2}   \rfloor} \choose {m - n}} + {{m - \lfloor \frac{n+2}{2} \rfloor}
      \choose {m - n}} = \Theta\big(m^{\lfloor n/2\rfloor}\big).$$
  \normalsize
\end{Thm}

The upper bound theorem, in its dual form, implies an upper bound on the number of vertices of a polytope with $m$ facets,
and the maximum is achieved at the dual of a cyclic polytope with $m$ vertices.

\begin{Thm}
  Any algorithm that solves the \ulp problem with $m$ constraints and $n$ variables 
  in the worst case needs at least $\Omega\big(m^{\lfloor n/2\rfloor}\big)$ queries to the oracle.
\end{Thm}

\begin{proof}
  By Theorem~\ref{thm:upperb} and the duality of polytopes, we know
  that the dual polytope $P$ of a cyclic polytope $C(n, m)$ has $m$
  facets and $k = \Theta\big(m^{\lfloor n/2\rfloor}\big)$ vertices. Equivalently, $P$ can be defined
  by those $m$ facet-defining half-spaces. Note that since the combinatorial structure of a cyclic polytope $C(n, m)$ is only determined by $m$ and $n$,
  the combinatorial structure of $P$ is also fixed.

  Assume that the vertices of polytope $P$ are $v_1, v_2, \ldots, v_k$. It is easy to see that for any $v_i$, $1
  \leq i \leq k$, there exists a half-space $H_i =\{y\in \mathbb{R}^n : \langle c_i, y \rangle \geq d_i\}$ which intersects $P$ only at $v_i$, i.e., $v_i = P \cap H_i$.
  Thus, $\langle c_i, v_i \rangle = d_i$ and $\langle c_i, v_j \rangle < d_i$ for all $j \neq i$.
  We now slightly move each half-space $H_i$ towards the polytope $P$ such that (i) $H_i\cap P$ has a positive volume, and (ii) $(H_i\cap P)\cap (H_j\cap P) = \emptyset$ for any $i\neq j$; denote the resulting open half-space by $H'_i$.


  We construct our unknown LP instances as follows. Let $P'$ be the set of interior points of $P$.
  Consider the following family of LP systems $LP_i: \big\{P' \cap H'_i\big\}$, for $i=1,\ldots,k$.
  From the above analysis, we know that
  \begin{itemize}[$\bullet$]
  \addtolength{\itemsep}{-0.5\baselineskip}
  \item every $LP_i$ consists of $n$ variables and $m+1$ constraints: the first $m$ ones by $P'$ and the last by $H'_i$;
  \item every $LP_i$ has a nonempty feasible region;
  \item the feasible regions of $LP_i$ and $LP_j$ are disjoint, $\forall i \neq j$.
  \end{itemize}

  Next we define the adversary oracle:
  For any queried point $p$, if $p \notin P'$,
  the oracle returns the index of an arbitrary violated constraint among the first $m$ constraints that define $P'$.
  If $p \in P'\setminus \bigcup_{i}H'_i$, the oracle just returns the index $m+1$, meaning that the last constraint (i.e., the one corresponds to $H'_i$) is violated.
  If $p\in P'\cap H_i'$ for some $i$,
  then the oracle also returns the index $m+1$ if the algorithm has not made $k$ queries yet.

  Now for any two systems $LP_i$ and $LP_j$, the oracle will return us a different answer only if we
  propose a point in the feasible region of one of them. Because overall there
  are $k$ LP systems, for any $\ell \leq k-2$ queries, there are at least two linear systems $LP_i$
  and $LP_j$ from which we cannot distinguish. Thus, an
  algorithm has to query on at least $k-1 = \Omega\big(m^{\lfloor n/2\rfloor}\big)$ candidates in
  order to solve the \ulp problem in the worse case. Hence, the theorem follows.
\end{proof}

\section{Concluding Remarks}

We consider solving linear programs when the input constraints are unknown,
and show that different kinds of violation information yield different computational complexities.
Linear programs are powerful tools employed in real applications dealing with objects that are largely unknown.
For example, in the node localization of sensor networks where the locations of targets are unknown~\cite{DPE01},
the computation of the locations in some settings can be formulated as
a linear program with constraints that
measure partial information obtained from data~\cite{Gen05}.
However, the estimation usually has various levels of error, which may lead to violations of the presumed constraints.
Interesting questions that deserve further explorations are what can
be theoretically analyzed there, and in general, what other natural
formats of violations there are in linear programming and what complexities they impose.

\section{Acknowledgments}

We thank Christos Papadimitriou for the helpful discussions and bringing~\cite{PY93} to our attention,
and Yingjun Zhang for providing a nice overview of transmit power control problems.

\newpage
\normalsize
\appendix

\section{Small Number of Constraints}\label{section-smallconstraint}

In this section, we consider those LP instances $Ax>b$ in which the number of constraints $m$ is less than or equal to the number of variables $n$.
Note that such LP instances possess an important property that a feasible solution almost surely exists.
In the following we will see that, in contrast to the exponential lower bound established for the general setting,
the solution-existing (almost everywhere) property yields an efficient algorithm for the \ulp problem.

\begin{Thm}
  The \ulp problem with $m$ constraints and $n$ variables can be solved in polynomial time with respect to the input size when $m\le n$.
\end{Thm}

\begin{proof}
  We use the ellipsoid method to find the unknown matrix $A$ and vector
  $b$, which can be together viewed as a point (also a degenerate polyhedron)
  in dimension $\mbR^{m(n+1)}$. Initially, we choose a sufficiently
  large ellipsoid that contains the candidate region of the point $(A,
  b)$. During each iteration of the algorithm, we pick the center
  $(A', b') \in \mbR^{m(n+1)}$ of the current ellipsoid.

  If the LP $A'x > b'$ has a feasible solution $x$, then we simply query
  $x$ to the oracle. If the oracle returns \yes, then $x$ is also a feasible solution to $Ax>b$, and the job is done.
  Otherwise, suppose that the oracle returns an index $i$, then we know that
  $\langle a_i, x \rangle < b_i$ and $\langle a'_i, x \rangle > b'_i$,
  which gives a separating hyperplane.

  A problem arises when $A'x > b'$ is infeasible, in which case we cannot find a separating
  hyperplane directly. However, notice that when $m \le n$, a linear
  program $A'x > b'$ is always feasible if the matrix $A'$ is full rank.
  Also, for any matrix $A'$ and any $\epsilon > 0$, we can easily find
  a full rank matrix $A''$ such that the difference between any entry of $A'$ and $A''$ is at most $\epsilon$.
  Then, by querying a feasible solution of $A''x > b'$, we actually get a
  separating hyperplane between $(A'', b')$ and $(A, b)$. Since
  $(A'', b')$ can be arbitrarily close to the center point $(A', b')$
  of the current ellipsoid, we are still be able to use the original ellipsoid argument to
  claim that the volume of the ellipsoid uniformly decreases at every step.

  Finally, notice that the solution polyhedron degenerates to a point and has volume 0.
  We can use a machinery developed by Gr{\"o}tschel, Lov{\'a}sz, and Schrijver~\cite{GLS84,GLS88} to handle this issue.
  The same idea has been used in solving \ulp in the furthest oracle model.

  Thus, the \ulp problem can be solved in polynomial time when $m\le n$.
\end{proof}

\newcommand\ctn{{\sf continue}}
\newcommand\term{{\sf terminate}}

\section{Furthest Oracle Algorithm: Formal Specifications and Proofs}\label{appendix-furthest-formal}

We give formal description of algorithm and proof of its correctness and complexity in this section. In Section \ref{appendix-lemma-furthest}, we show a proof of Lemma~\ref{lemma-furthest}. In Section \ref{sec:cons-check}, we give a procedure to check the consistency of two furthest Voronoi diagrams, which will be used in the main algorithm given in Section \ref{sec:bMainAlg}.

\subsection{Proof of Lemma~\ref{lemma-furthest}}\label{appendix-lemma-furthest}

\begin{proof}
  On the one direction, assume that $Ax > b$ is infeasible. For any point $p\in \extreme(A, b)$, by definition, $p$ is not a feasible solution of $Ax > b$. Suppose for the sake of contradiction that the support linear system $\{\langle a_i, x \rangle > 0 \mid i \in S\}$ contains a feasible solution $x \in \mbR^n$. Consider $x' = p + \epsilon x$ for a small $\epsilon>0$. We have
  \[
  \max_i\big\{b_i - \langle a_i, x' \rangle\big\} = \max_i\big\{b_i - \langle a_i, p \rangle - \epsilon\cdot \langle a_i,x \rangle\big\} < d(A,b).
  \]
  This contradicts the fact that $p\in \extreme(A, b)$. Thus, we know that $\{\langle a_i, x \rangle > 0 \mid i \in S\}$ is infeasible.

  On the other direction, let $$d' = \max_{i\in [m]} \{b_i - \langle a_i, p \rangle\}.$$ Since $p$ is not a feasible solution of $Ax>b$, we have $d'\ge 0$.
  By the definition of $S$, we know that for any $i \in S$, $b_i - \langle a_i, p
  \rangle = d' \geq 0$.
  Now consider any point $x \in \mbR^n$, let $x' = x-p$.
  Since there exists an index $i \in S$ such that $\langle a_i, x'
  \rangle \leq 0$, we have
  \[
    b_i - \langle a_i, x \rangle  =  b_i - \big\langle a_i, p \big\rangle - \big\langle a_i, x' \big\rangle \geq d' + 0 \geq 0
  \]
  This means that $x$ is not a feasible solution to $Ax > b$. Thus, we conclude that $Ax > b$ is infeasible.
\end{proof}

\subsection{Consistency Check between Voronoi Diagrams}\label{sec:cons-check}

To specify the main algorithm, we need a subprocedure \SameVor to check whether the furthest Voronoi diagram of the unknown $(A,b)$ is the same as that of a proposed $(A',b')$. 
We will describe the procedure with respect to the general $Ax>b$ case, and the discussions in Section~\ref{sec:Ax>0} for $Ax>0$ can be considered as a special case.
(Recall that the input instance has $m$ constraints, $n$ variables, and binary size $L$.)

\begin{algorithm}[H]
  \floatname{algorithm}{Procedure}
  \caption{\SameVor$(A', b')$}
  {\bf Input:} Presumed point $(A', b')$ \\
	{\bf Output:} \yes, or a separating hyperplane (between $(A,b)$ and $(A',b')$)
  \begin{algorithmic}[1]
    \FOR {each pair $i,j\in [m]$ and $i \neq j$}
      \IF {$\gvor'(i)\cap \gvor'(j) \neq \emptyset$}
        \STATE Solve the following LP feasibility problem and get a solution $y^{(1)}$:
        \begin{align}
          \label{eq:sameVor1} b_i' - \langle a_i', y\rangle = b_j' - \langle a_j', y\rangle &> b_k' - \langle a_k', y\rangle + 2^{-10L}, \forall k\neq i,j
				\end{align}
        \STATE Take an orthonormal basis $z^{(2)}, \ldots, z^{(n)}$ of the \mbox{affine subspace $H'_{ij}\supseteq \gvor'(i)\cap \gvor'(j)$.}
        \FOR {each $k \in \{2, ..., n\}$}
          \STATE Let $y^{(k)} = y^{(1)} + 2^{-10L} \cdot z^{(k)}$.
          \STATE Find $\epsilon^{(k)}$ (the absolute value of each of its component is bounded by $2^{-2L}$) to satisfy the following inequalities Eq.\eqref{eq:sameVor2} and \eqref{eq:sameVor3}.
          \begin{align}
						\label{eq:sameVor2} b_i' - \langle a_i', y^{(k)} + \epsilon^{(k)} \rangle &> b_\ell' - \langle a_\ell', y^{(k)} + \epsilon^{(k)} \rangle , \forall \ell\neq i \\
						\label{eq:sameVor3} b_j' - \langle a_j', y^{(k)} - \epsilon^{(k)} \rangle &> b_\ell' - \langle a_\ell', y^{(k)} - \epsilon^{(k)} \rangle , \forall \ell\neq j
					\end{align}
        \ENDFOR
        \FOR {each $k \in [n]$}
          \STATE Query $y^{(k)} + \epsilon^{(k)}$ and $y^{(k)} - \epsilon^{(k)}$ to the oracle, and get answers $i'$ and $j'$, respectively.
          \IF {$i' \neq i$ or $j' \neq j$}
						\STATE Output the corresponding separating hyperplane.
					\ENDIF
        \ENDFOR \label{step:intersection}
      \ENDIF
    \ENDFOR
    \STATE Output \yes
  \end{algorithmic}
\end{algorithm}

\begin{Lem}\label{lem:const-check}
  If $\gvor'(i)\neq \emptyset$ for all $i$, then the procedure \SameVor either finds a separating hyperplane between $(A', b')$ and $(A, b)$, or confirms that $\gvor = \gvor'$.
\end{Lem}

\begin{proof}
  In the procedure \SameVor, for each $(i,j)$ with $\gvor'(i)\cap \gvor'(j) \neq \emptyset$, we first find $n-1$ linearly independent points in $\gvor'(i)\cap \gvor'(j)$, and then check that they are also in $\gvor(i)\cap \gvor(j)$. We achieve this by finding a pair of points $y^{(k)}\pm \epsilon^{(k)}$ and verify that they are in $\gvor(i)$ and $\gvor(j)$, respectively. Then because of the assumed precision for $(A,b)$, we know that $y^{(k)} \in \gvor(i)\cap \gvor(j)$. Hence, after Step \ref{step:intersection} in \SameVor$(A',b')$, we can conclude that the hyperplane $H'_{ij}$ also contains $\gvor(i)\cap \gvor(j)$.
	
Now consider each nonempty $\gvor'(i)$. It is the intersection of $k\leq m-1$ half-spaces with boundary $H'_{ij}$, and each of these $H'_{ij}$'s is equal to the corresponding boundary hyperplane $H_{ij}$ in $\gvor$. Thus, we know that $\gvor(i) \subseteq \gvor'(i)$ as the former is the intersection of possibly more half-spaces. But both $\gvor$ and $\gvor'$ are a partition of $\mbR^n$, together with the fact that all $\gvor'(i)$'s are nonempty, we have $\gvor = \gvor'$.
\end{proof}

The procedure \SameVor$(A',b')$ checks whether $\gvor = \gvor'$, assuming that $\gvor'(i) \neq \emptyset$ for all $i$. However, if $\gvor'(i) = \emptyset$ for some $i$, but \SameVor does not know whether $\gvor(i)$ is empty. Our solution to this issue is to simply ignore these indices $i$ that are not in $T$ and pretend that the unknown LP does not have these constraints.
If at any point the oracle outputs some $i$ that is not in $T$ (i.e., with $\gvor'(i) = \emptyset$), then we get a separating hyperplane. Otherwise the algorithm just runs as if these $i$'s do not exist: We either find some other separating hyperplane or a valid solution to the unknown LP, or confirm that the unknown LP restricted to $T$ is infeasible, which implies that the original LP is also infeasible (since it needs to satisfy even more constraints).

\subsection{Main Algorithm}\label{sec:bMainAlg}

Now we are ready to describe the main algorithm \FurthestAlg, which uses the ellipsoid method to search for $(A,b)\in \mbR^{m(n+1)}$ in order to solve the general \ulp problem $Ax>b$.
In the algorithm, whenever we find a hyperplane that separates the unknown $(A,b)$ and the center $(A',b')$ of the current ellipsoid, we use ``\ctn'' to denote the standard procedure of proceeding to the next iteration of the ellipsoid method with a smaller ellipsoid. Whenever we find a solution $x$ with $Ax > b$, we use ``$\term(x)$'' to mean to terminate the whole program with an output $x$. We also use ``$\term(\no)$'' for terminating the whole program with $\no$, i.e., no feasible solution exists. We use $B(p,\epsilon)$ to denote the ball centered at $p$ with radius $\epsilon$.



\begin{algorithm}[H]
  \caption{\FurthestAlg}
  {\bf Input:} an unknown LP $Ax>b$ \\
  {\bf Output:} a valid solution $x$ for $Ax>b$, or \no.

  \begin{algorithmic}[1]
    \STATE Run the ellipsoid method to search for $(A,b)\in \mbR^{m(n+1)}$.
    \FOR {each iteration of the ellipsoid method}
      \STATE Let $(A',b')$ be the center of the current ellipsoid.
      \IF {$A'x > b'$ is feasible}
        \STATE Find a solution $x$ to it and query $x$ to the oracle.
        \STATE {\bf if} the oracle returns \yes \ {\bf then} $\term(x)$.
        \STATE {\bf else} suppose that the oracle returns $i$, then we get a separating hyperplane \\ \quad ``$a_i x \leq b_i$ and $a_i' x > b'_i$'' and \ctn.
      \ELSE 
        \STATE \label{step:def T} Define $\gvor$ and $\gvor'$ by Eq.\eqref{eq:GV} and \eqref{eq:GV'}, and let $T = \{i : \gvor'(i) \neq \emptyset\}$.
				\STATE \label{step:subproc} Run \SameVor$(A',b')$. If it outputs a separating hyperplane, then \ctn
				\STATE Compute a point $p\in \extreme(A',b')$ and the corresponding support $S$.
        \STATE \label{step:feasible} Query $p$ to the oracle. 
				\IF {the oracle outputs \yes}
                    \STATE $\term(p)$.
				\ELSIF {the oracle outputs an $i\notin T$}
						\STATE we get a separating hyperplane and \ctn.
				\ELSE
					\STATE Use $\gvor = \gvor'$ in $B(p,\epsilon)$ for a sufficiently small $\epsilon$ to recover $\{\langle a_i, x \rangle > 0 : i\in S\}$.
					\IF {$\{\langle a_i, x \rangle > 0 : i\in S\}$ is feasible}
						\STATE Compute a feasible solution $x^*$, and suppose $\langle a'_i, x^* \rangle \leq 0$ for some $i \in S$. We get a separating hyperplane ``$\langle a'_i, x^* \rangle \leq 0$ and $\langle a_i, x^* \rangle > 0$'' and \ctn.
					\ELSE
						\STATE $\term(\no)$.
					\ENDIF
				\ENDIF
      \ENDIF
    \ENDFOR
  \end{algorithmic}
\end{algorithm}

\begin{Thm}\label{thm:furthest}
Algorithm \FurthestAlg solves the \ulp problem $Ax>b$ with the furthest oracle in polynomial time.
\end{Thm}

\begin{proof}
We will prove the correctness and analyze the complexity along the way. The whole algorithm runs the ellipsoid method to search for $(A,b)\in \mbR^{m(n+1)}$. For the center $(A',b')$ of the current ellipsoid, if the LP $A'x > b'$ is feasible, then we can find a solution $x$ to $A'x > b'$ in polynomial time. If this $x$ is also a solution to the unknown LP $Ax>b$ from the query to the oracle, then the algorithm successfully finds a feasible solution; thus, it outputs $x$ and terminates. If the oracle returns some $i$, it means that $a_ix \leq b_i$, then together with $a_i' x > b'_i$, we get a hyperplane that separates $(A,b)$ and $(A',b')$ in $\mbR^{m(n+1)}$. Thus, we can go to the next iteration of the ellipsoid method.
Therefore, the hard case is when the LP $A'x > b'$ is infeasible, which will be our assumption for the rest of the proof.

Define $\gvor$ and $\gvor'$ by Eq.\eqref{eq:GV} and \eqref{eq:GV'}. The procedure \SameVor$(A',b')$ checks whether $\gvor = \gvor'$, assuming that $T = \{i : \gvor'(i) \neq \emptyset\}=[m]$. By Lemma \ref{lem:const-check}, if $T = [m]$, then either the procedure finds a separating hyperplane for $(A,b)$ and $(A',b')$, or confirms that $\gvor = \gvor'$. If $T \neq [m]$, then some $\gvor'(i) = \emptyset$, but \SameVor does not know whether $\gvor(i)$ is empty too.
Our solution to this issue is to simply ignore these indices $i$ not in $T$ and pretend that the LP does not have these constraints. 
If at any point (of the current ellipsoid iteration, i.e., for the current $(A',b')$), the oracle outputs some $i$ that is not in $T$ (i.e., with $\gvor'(i) = \emptyset$), then we get a separating hyperplane (because in $(A',b')$ the distance $b'_i - \langle a'_i, x\rangle$ is always smaller than $b'_j - \langle a'_j, x\rangle$ for some $j$). If in other steps of the algorithm, i.e., at Step 10 of \SameVor$(A',b')$ and Step 16 of \FurthestAlg, the oracle never returns an $i\in [m]\setminus T$, then the algorithm just runs as if these $i$'s do not exist: We either find some other separating hyperplane or a valid solution to the unknown LP in some other steps, or we confirm that the unknown LP restricted to $T$ is infeasible, which implies that the original LP is also infeasible (since it needs to satisfy even more constraints).

	Next we show that one can compute a point $p\in \extreme(A',b')$ and the corresponding support $S$ in polynomial time, which are defined by Eq.\eqref{eq:vertexdist} and \eqref{eq:extreme}. From Eq.\eqref{eq:vertexdist}, we see that when $A'x>b'$ is infeasible, $d(A',b') \geq 0$. Since Eq.\eqref{eq:vertexdist} can be expressed as an LP ($\min z$ \st $b_i - \langle a_i, x \rangle \leq z$, $\forall i\in [m]$), we know that the minimum is always achievable and can be computed in polynomial time. Then, from Eq.\eqref{eq:extreme}, one can search over all $i\in [m]$ for an $x$ satisfying that $b_i - \langle a_i, x \rangle = d(A',b') \geq b_j - \langle a_j, x \rangle, \ \forall j\neq i$. Finally, it is easy to fix $S$ as the set of indices $i$ with $b_i - \langle a_i, x \rangle = d(A',b')$.
	
	Now that we have found a point $p\in \extreme(A',b')$; if $p$ happens to be a solution of $Ax>b$, then we are done. Below we focus on the situation that $Ap > b$ does not hold. By Lemma~\ref{lem:ax>b}, it suffices to show that the support linear system of $Ax > b$ at $p$ is infeasible. Let $S$ (and $S'$, respectively) be the set of the indices of the half-spaces in $Ax > b$ (and $A'x > b'$, respectively) that have the (same) maximal distance to point $p$. Since $\gvor=\gvor'$ and the set of indices that have the (same) maximal distance to some point solely depends on the Voronoi diagram, we know $S = S'$.
Consider the corresponding support linear system $\{\langle a_i, x \rangle > 0 : i \in S\}$ for the unknown LP $Ax>b$. 
We pick a small enough $\epsilon^* > 0$ such that for any $i \notin S$, $\epsilon^* < d - (b_i - \langle a_i, p \rangle)$. (Notice that though we do not know the hidden $(A,b)$, the assumed precision implies a minimum possible gap between the largest distance $d(A,b)$ and the second largest distance $\max_{i: b_i-\langle a_i, p \rangle < d(A,b)} b_i-\langle a_i, p \rangle$. It is not hard to see that this gap is singly exponentially small, thus we can take an $\epsilon$ smaller than this gap using polynomial number of bits.)

Now, if we focus on the region around $p$ and limit our queries within the ball $B = \{z\in \mathbb{R}^n :  \|z - p\| \leq \epsilon^*\}$, then for each query $x = p + x'$ within the ball $B$, where $\|x'\|\le \epsilon^*$, the oracle returns an index
\begin{eqnarray*}
i &\in & \arg\max_{i\in [m]}\big\{b_i - \langle a_i, x \rangle\big\} = \arg\max_{i\in S}\big\{b_i - \langle a_i, x \rangle\big\} \\
  &=& \arg\max_{i \in S}\big\{b_i - \langle a_i, p \rangle -\langle a_i, x' \rangle\big\} = \arg\max_{i \in S}\big\{-\langle a_i, x' \rangle\big\},
\end{eqnarray*}
where the first equality is because when $\epsilon^*$ is sufficiently small, the maximum is always achieved by some $i\in S$, by the definition of $S$. The last equality follows from the fact that all half-spaces in $S$ have the same maximal distance to $p$.
Thus, for any queried point within this ball $B$, the set of indices possibly returned by the oracle for $Ax > b$ is the same as that by the furthest oracle for the support linear system $\{\langle a_i, x \rangle > 0 :  i\in S\}$.
This means that the Voronoi structure $\gvor$ in $B$ is exactly the same as the weighted spherical\footnote{Recall the assumption of $\|a_i\| = 1$, for all $i\in [m]$, for both $Ax>0$ in Section \ref{sec:Ax>0} and $Ax>b$ in Section \ref{sec:Ax>b}.} Voronoi diagram $\gvor'$ for the corresponding support system in $B$. Note that the non-degenerency assumption of $(A,b)$ holds in the ball $B$ as well. Thus, by the results in Section \ref{sec:Ax>0}, we are able to recover all $a_i$ for $i\in S$; namely the support system $\{\langle a_i, x \rangle > 0 : i\in S\}$ is revealed.
The last step is straightforward: if the support system is infeasible, then by Lemma~\ref{lem:ax>b}, we can conclude that LP $Ax > b$ is also infeasible. If the support system has a feasible solution $x^*$, since we know $x^*$ is not a feasible solution to $\{\langle a'_i, x \rangle > 0 : i\in S\}$ (by Lemma~\ref{lem:ax>b}), we must have $\langle a'_i, x \rangle \leq 0 < \langle a_i, x \rangle$ for some $i \in S$, which is a separating hyperplane between $(A, b)$ and $(A', b')$.

From the above discussions, we know that every step can be implemented
in polynomial time. Hence, by the ellipsoid method, we can solve the problem efficiently.
\end{proof} 

\newcommand{\refl}{{\sf ref}\xspace}
\newcommand\cl[1]{#1}

\section{Proof of the Chamber Counting Theorem}\label{appendix-chamber}

In this section, we will give a formal proof of Theorem~\ref{thm:chambers} for counting the number of chambers in the general $n$-dimensional space.
We first define direction vectors in $\mbR^n$ and their indicator functions;
these definitions are inspired by Banchoff's seminal work~\cite{Ban67}.

\begin{Def}\label{def:general}
  For a given closed set $C$ in $\mathbb{R}^n$, a vector $v \in \mbR^n$ is called {\em general} if there is $x\in C$ such that $\langle
  v, x \rangle > \langle v, y \rangle$ for any $y \in C$ and $y\neq x$. That is, $x$ is the unique extreme point in $C$ along the direction $v$.
  Note that $x$, if existing, must be at the boundary of $C$, denoted by $\partial C$.
  Define an indicator function $f_v(p,C)$ by $f_v(p,C)=1$ if $p$ is such a unique extreme point, and $f_v(p,C) = 0$ otherwise.
\end{Def}

By the definition, it was shown in~\cite{Ban67} that if $v$ is general for a bounded set $C$, then $\sum_{p \in \partial C}f_v(p, C) = 1$.
Let $d\omega^{n-1}$ be the ordinary volume element on the sphere of the unit ball in $\mathbb{R}^n$, denoted by $S^{n-1} = \{p \in \mbR^n : \|p\| = 1\}$
and let $V = \int_{S^{n-1}}d\omega^{n-1}$ be the volume of $S^{n-1}$.
Note that the integral $\int_{S^{n-1}}$ can be considered as either over all points in $S^{n-1}$ or over all direction vectors, i.e., $\int_{v\in S^{n-1}}$.
Then for any nonempty bounded convex set $C \subset \mbR^n$, it holds that
\begin{equation}\label{eq:sum=1}
V = \int_{S^{n-1}}d\omega^{n-1} = \int_{S^{n-1}} \sum_{p \in \partial C}f_v(p,C) \, d\omega^{n-1} = \sum_{p \in \partial C} \int_{S^{n-1}} f_v(p,C) \, d\omega^{n-1}.
\end{equation}
That is, the summation of the indicator functions over all points on the boundary of $C$ over all directions equals to the volume of $S^{n-1}$.

\medskip \noindent
\textbf{Bound the unbounds.}
Note that Equation~(\ref{eq:sum=1}) does not hold when the set $C$ is unbounded. The major effort in our proof is devoted to dealing with unbounded convex sets, described in the following.
We add a ball $B(r) = \{p\in \mbR^n : \|p\| \leq r\}$ with radius $r$ sufficiently large to satisfy the following conditions.
\begin{itemize}
\item All bounded convex sets are contained within $B(r)$.
\item At any intersection point between the boundary of an unbounded
  convex set and the sphere of $B(r)$, a supporting hyperplane of the
  convex set is ``almost'' perpendicular to the tangent hyperplane of $B(r)$ at that point. Formally,
  for any point $p\in \partial C_i \cap \partial B(r)$ for some unbounded convex set $C_i$, let
  the unique supporting hyperplane of $B$ at $p$ be $\{x\in \mathbb{R}^n : \langle p, x \rangle = b\}$, then 
  the distance between $(1-\epsilon)p$ and the supporting hyperplane of $C_i$ is 0 when $r$ approaches infinity.
  This implies that for any $x$ which is not linear to $p$, if $\lim\limits_{r\to\infty} p+x \in C_i\cap B(r)$, then $\lim\limits_{r\to\infty}
  p+\refl(x,p)$ is not an interior point of $C_i\cap B(r)$ (see Figure~\ref{figure-chamber-b}),
  where $\refl(x,p)=(2pp^T-I)x$ reflects $x$ across the line that goes through the origin and $p$ and $I$ is the identity matrix.
\end{itemize}

Given the big ball $B(r)$ described as above, let $C'_i$ be the
closure of set $C_i \cap B(r)$ for all $1 \leq i \leq m$; note that
all sets $C'_i$ are now closed and bounded.
Let $$D' = \Big(\mbR^n \backslash \bigcup_iC_i\Big) \cap B(r)$$ be the complement of the union of
these convex sets within the ball $B(r)$. Note that the number of chambers in $D'$ and that in $D = \mbR^n \backslash \bigcup_iC_i$ is the same.
Suppose that $D'$ has $k$ chambers, and let $D'_1, \ldots, D'_k$ be the
closure of each chamber (which are all bounded).\footnote{Note that the definitions of all $C'_i$ and $D'_j$
depend on the radius of the ball $B(r)$. That is, precisely, they should be $C'_i(r)$ and $D'_j(r)$. In our discussions, for simplicity, we
ignore the parameter $r$, and note that $r$ is always sufficiently large.}
Next we partition all points in $\bigcup_j{\partial D'_j}$ into three categories:
\begin{itemize}
\item $\Gamma_B = \big\{p \in \bigcup_j{\partial D'_j} ~|~ p \in \partial B, p \notin \bigcup_i{\partial C'_i}\big\}$.
\item $\Gamma_C = \big\{p \in \bigcup_j{\partial D'_j} ~|~ p \notin \partial B, p \in \bigcup_i{\partial C'_i}\big\}$.
\item $\Gamma_{BC} = \big\{p \in \bigcup_j{\partial D'_j} ~|~ p \in \partial B, p \in \bigcup_i{\partial C'_i}\big\}$.
\end{itemize}

We next prove a key lemma used in our proof. Note that we can assume that any $D'_j$ and $D'_{j'}$, $j\neq j'$, do not intersect; indeed, we can expand all the closed sets $C_i$ by a small $\epsilon$ and this does not decrease the number of chambers. 

\begin{Lem} \label{lem:cham:ineq}
  For any chamber $D'_j$, any $p \in \partial D'_j$, and any vector $v\in \mathbb{R}^n$, the following properties hold.
  \begin{itemize}
  \item [(1)] If $p \in \Gamma_B$, $f_v(p, \cl{D'_j}) = f_v(p, B)$.
  \item [(2)] If $p \in \Gamma_C$, $f_v(p, \cl{D'_j}) \leq f_{-v}\Big(p, \bigcap_{i:p \in \partial C'_i} \cl{C'_i}\Big)$.
  \item [(3)] If $p \in \Gamma_{BC}$, $\lim\limits_{r\to\infty}f_v\big(p, \cl{D'_j}\big) \leq \lim\limits_{r\to\infty}f_{\refl(v,p)}\bigg(p, \bigcap_{i:p \in \partial C'_i} \cl{C'_{i}}\bigg).$
  \end{itemize}
\end{Lem}

\begin{proof}
  We prove the claim for each case respectively.
  \begin{itemize}
  \item [(1)] Both sides equal to 1 if $v=\lambda p$ for some $\lambda > 0$, and both are 0 otherwise.

	\item [(2)] Let $I= \{i\in [m] ~|~ p\in \cl{C'_i}\}$.  If the right hand side of the inequality is 0, then there is another point $w \in \bigcup_{i\in I} \cl{C'_i}$ with
    $\langle p, -v \rangle < \langle w, -v \rangle$. Let $w = p + x$. We pick a
    small enough $\epsilon > 0$, such that $p' = p - \epsilon x \notin \cl{C'_j}$ for any $j\notin I$ (see Figure~\ref{figure-chamber-a}).
    At the same time, for any $i \in I$, $p' \notin \cl{C'_i}$ as
    well, because there is a separating hyperplane such that $\cl{C'_i}$ is
    entirely at one side, and thus, $p'$ cannot be an interior point of
    $C'_i$. So $p'\in \cl{D'_j}$ as we assumed that two $\cl{D'_j}$'s
    do not intersect. Now $\langle p', v \rangle = \langle p, v
    \rangle - \epsilon (\langle p, v \rangle - \langle w, v \rangle) > \langle p, v \rangle$, thus, $f_v(p, \cl{D'_j}) = 0$.


    \begin{figure}[h]
    \centering
    \subfloat[{\small Point $p\in \Gamma_c$.}]{\label{figure-chamber-a}\includegraphics[width=0.45\textwidth]{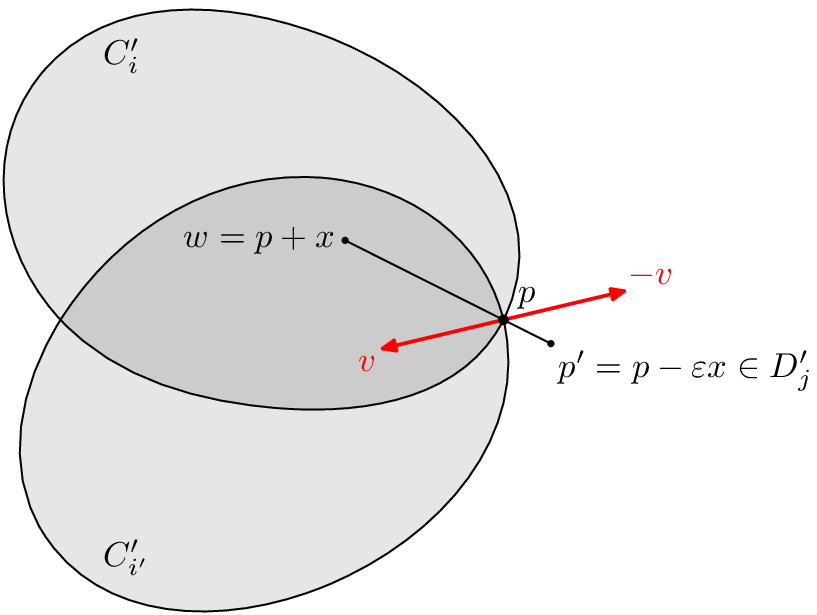}}
    \vspace{0.1in}
    \subfloat[{\small Point $p\in \Gamma_{BC}$.}]{\label{figure-chamber-b}\includegraphics[width=0.45\textwidth]{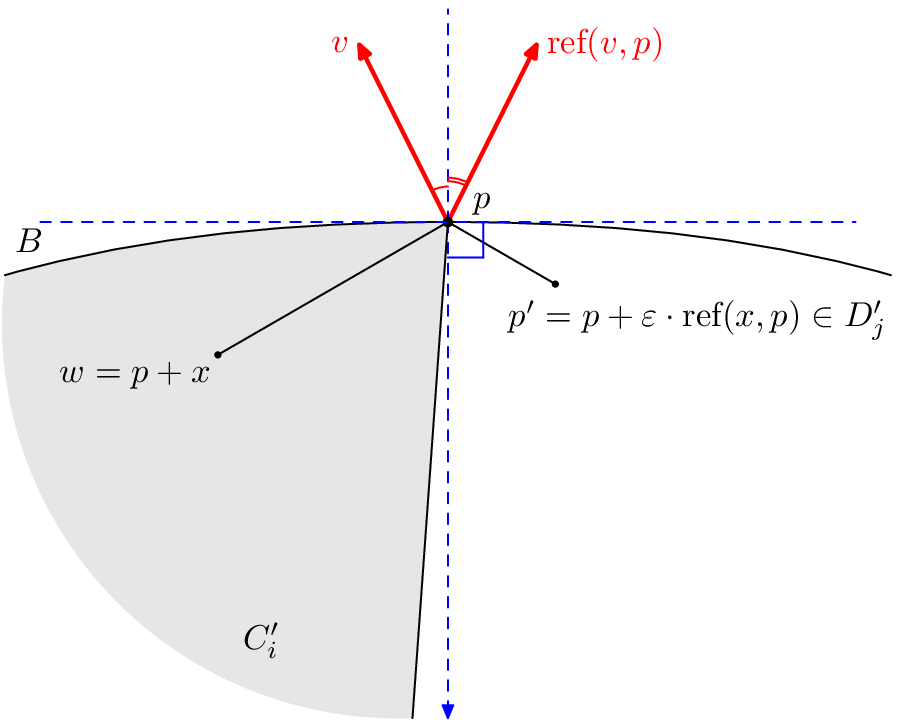}}
    \caption{Illustration of the boundaries.}
    \end{figure}

	\item [(3)] We use a similar argument as above. If the
          right hand side of the inequality is 0, then there is another point $w \in \bigcup_{i\in I} \cl{C'_i}$ with
    $\langle p, \refl(v, p) \rangle < \langle w, \refl(v, p) \rangle$. Let $w = p + x$. We pick a
    small enough $\epsilon > 0$, such that $p' = p +
    \epsilon(\refl(x, p)) \notin \cl{C'_j}$ for any $j\notin I$ (see Figure~\ref{figure-chamber-b}).
    At the same time, when $r$ is large enough, for any $i \in I$, because $p + \epsilon
    x \in C'_i$, which implies that $p'$ is not an interior point of $C'_i$. Then we know $p'\in
    \cl{D'_j}$. Thus, $\langle p', v \rangle = \langle p, v \rangle +
    \epsilon \langle \refl(x, p), v \rangle = \langle p, v
    \rangle + \epsilon \langle \refl(v, p), x \rangle > \langle p, v \rangle$, which
    implies $f_v(p, \cl{D'_j}) = 0$.


  \end{itemize}
  This completes the proof of the lemma.
\end{proof}

In addition, notice that for any point $p$, direction vector $v'$ and bounded set $C$, we have
\begin{eqnarray*}
\int_{S^{n-1}}f_{v}(p, C)\, d\omega^{n-1} = \int_{S^{n-1}}f_{-v}(p, C)\, d\omega^{n-1} = \int_{S^{n-1}}f_{\refl(v,v')}(p, C)\, d\omega^{n-1}.
\end{eqnarray*}
This fact, together with the above lemma, implies the following important corollary.
\begin{Cor} \label{cor:cham:ineq}
  For any chamber $D'_j$ and $p \in \partial D'_j$, the following hold.
  \begin{itemize}
  \item [(1)] If $p \in \Gamma_B$, then $\int_{S^{n-1}}f_v(p, D'_j) \, d\omega^{n-1} = \int_{S^{n-1}}f_v(p, B) \, d\omega^{n-1}$.
  \item [(2)] If $p \in \Gamma_C \cup \Gamma_{BC}$, then
        $\int_{S^{n-1}}\lim\limits_{r\to\infty}f_v(p, D'_j) \, d\omega^{n-1}
        \leq \int_{S^{n-1}}\lim\limits_{r\to\infty}f_v\bigg(p, \bigcap_{i:v \in \partial C'_i}C'_i\bigg) \, d\omega^{n-1}.$
  \end{itemize}
\end{Cor}

Before proving the main theorem, we need to handle degenerated cases in which some points are at the boundary of more than $n$ convex sets.

\begin{Lem} \label{lem:cham:dege}
  For any point $p\in \mathbb{R}^n$, let $S = \big\{C_i ~|~ 1 \leq i \leq m, p \in \partial C_i\big\}$ be the set of convex sets whose boundary contains $p$. If
  $|S| > n$, then for any vector $v$ that is general to all these convex sets,
  $$f_v\bigg(p, \bigcap_{C \in S}C\bigg) \leq \sum_{S' \subset S, |S'| = n}f_v\bigg(p, \bigcap_{C \in S'}C\bigg).$$
\end{Lem}

\begin{proof}
  For notational convenience, we shift the origin of the coordinate system to $p$. An important observation is
  Note that for any convex set $C$ with $p\in \partial C$, $f_v(p, C) = 1$ if and only if $v$ is contained in the polar
  cone of set $C$. Thus, by Lemma~\ref{lem:polar}, we know that $f_v\big(p, \bigcap_{C \in S}C\big) = 1$ if any only if $v \in \big(\bigcap_{C \in S}C\big)^* = conv\big(\bigcup_{C \in S}C^*\big)$. Combining this with the fact
  that for any set $S$ of more than $n$ cones in $\mbR^n$, $conv(S) =
  \bigcup_{S' \subset S, |S'| = n}conv(S')$, we know that $f_v\big(p, \bigcap_{C \in S}C\big) = 1$ then there must exist
  $S' \subset S$, $|S'| = n$, such that $f_v\big(p, \bigcap_{C \in S'}C\big) = 1$. 
\end{proof}

Now the main theorem can be proved by a simple counting argument.

\begin{proof}[Proof of Theorem~\ref{thm:chambers}]
  As discussed above, assume there are $k$ chambers caused by $T\triangleq\{C'_1,\ldots,C'_m\}$ (which are all bounded). Then we have
  \small
  \begin{eqnarray*}
    k & = & \frac1V \sum_{j = 1}^k \sum_{p \in \partial D'_j} \int_{S^{n-1}}f_v(p, D'_j) \, d\omega^{n-1} \qquad \text{(Equation~\ref{eq:sum=1})} \\
    & = & \frac1V \bigg( \sum_{p \in \Gamma_B} + \sum_{p \in \Gamma_C} + \sum_{v \in \Gamma_{BC}} \bigg) \int_{S^{n-1}}f_v(p, D'_j) \,  d\omega^{n-1} \qquad \text{($j$ is the unique one with $p\in \cl{D'_j}$)}\\
    & \leq & \frac1V \sum_{p \in \partial B} \int_{S^{n-1}}f_v(p, B) \, d\omega^{n-1} +
    \frac1V \sum_{p \in \bigcup_j\partial D'_j} \int_{S^{n-1}}\lim_{r\to\infty}f_v\bigg(p, \bigcap_{i:p \in \partial C'_i}C'_i\bigg) \, d\omega^{n-1} \qquad
    \textrm{(Corollary~\ref{cor:cham:ineq})} \\
    & \leq & \frac1V \sum_{p \in \partial B} \int_{S^{n-1}}f_v(p, B) \, d\omega^{n-1} +
    \frac1V \sum_{\substack{S\subseteq T\\ 0 < |S|\le n}} \sum_{p \in \bigcap_{C'\in S} \partial C'} \int_{S^{n-1}}\lim_{r\to\infty}f_v\bigg(p, \bigcap_{C'\in S} C'\bigg) \, d\omega^{n-1} \qquad
    \textrm{(Lemma~\ref{lem:cham:dege})} \\
    & = & 1 + \sum_{\substack{S\subseteq T\\ 0 < |S|\le n}} 1 \qquad \textrm{(Equation~\ref{eq:sum=1})} \\
    & = & \sum_{i = 0}^n{m \choose i}
  \end{eqnarray*}
  \normalsize
  Therefore, the theorem follows.
\end{proof}

\end{document}